\documentclass[sn-apa]{sn-jnl}
\usepackage{marginnote}
\usepackage{amsfonts}
\usepackage{amsmath}
\usepackage{hyperref}
\usepackage{tabularx}
\usepackage{caption}
\usepackage{graphicx}
\usepackage{rotating}
\usepackage{amssymb}
\usepackage[running]{lineno}

\jyear{2021}%

\theoremstyle{thmstyleone}%
\newtheorem{theorem}{Theorem}
\newtheorem{proposition}[theorem]{Proposition}%

\theoremstyle{thmstyletwo}%
\newtheorem{remark}{Remark}%

\theoremstyle{thmstylethree}%
\usepackage{natbib}
 \newtheorem{lem}{Lemma}[]

\DeclareMathOperator*{\argmax}{arg\,max}

\DeclareMathOperator*{\T1}{\top}

\usepackage{latexsym}
\usepackage{graphbox}

\begin{document}

\title[Classification of multivariate functional data with PLS approaches]{Classification of multivariate functional data on different domains with Partial Least Squares approaches}


 \author*[1,2]{\fnm{Issam-Ali} \sur{Moindjié}}\email{issam-ali.moindjie@inria.fr}

 \author[1,2]{\fnm{Sophie} \sur{Dabo-Niang}}\email{sophie.dabo@univ-lille.fr}

  \author[1,2,3]{\fnm{Cristian} \sur{Preda}}\email{cristian.preda@univ-lille.fr}




\affil[1]{
\orgdiv{MODAL team}, \orgname{Inria
de l'Université de Lille}, \orgaddress{  \country{France}}}

\affil[2]{
\orgdiv{University of Lille}, \orgname{CNRS, UMR 8524 }, \orgaddress{ \city{Laboratoire Paul Painlevé}, \postcode{F-59000, France}}}

\affil[3]{
\orgdiv{Institute of Statistics and Applied Mathematics of the Romanian Academy}, \orgaddress{ \city{Bucharest}, \postcode{050711, Romania}}}



\abstract{Classification (supervised-learning) of multivariate functional data is
considered when the elements of the random functional vector of interest are defined on different domains.

In this setting, PLS classification and tree PLS-based methods for multivariate functional data are presented.
From a computational point of view, we show that the PLS components of the regression with multivariate functional data can be obtained using only the PLS methodology with univariate functional data.
This offers an alternative way to present the PLS algorithm for multivariate functional data.

Numerical simulation and real data applications highlight the performance of the proposed methods.}

\keywords{multivariate functional data analysis, supervised learning,  classification, partial least squares regression (PLS)}



\maketitle

\section{Introduction}\label{sec1}
In many areas, high-frequency data are monitored in time and space. For example, (i) in medicine, a patient's state can be diagnosed by time-related recordings (e.g. electroencephalogram, electrocardiogram) 
or/and images (e.g. fMRI) and (ii) in finance,  the stocks markets are naturally recorded in time and space. Analyzing such data requires adapted techniques, mainly because of the high dimension and their complex time and space correlation structure. 
Since the pioneer works of \cite{ramsay2008}, this data is well-known in statistics as functional data and, nowadays, it is a well-established statistical research domain. Viewed as a sample of a random variable with values in some infinite dimensional space, functional data is mostly associated with a random variable indexed by a continuous parameter such as the time, wavelengths, or percentage of some cycle.\par
Dimension reduction techniques are used in order to tackle the issue of high dimension and correlation. Among these, the most basic and  elementary one is the selection of privileged features of data by expert's knowledge (see e.g \cite{saikhu2019}, \cite{javed2020}). Some other works focused on deep learning models, in particular, long short-term memory models have been proposed for time series (\cite{LSTM}, \cite{LSTM_2}, \cite{multivariate}). They have the advantage of being less dependent on prior knowledge but are usually not interpretable. Maybe the most used methodologies for dealing with functional data are based on building latent models such as principal component analysis/regression (PCA, PCR) (\cite{ramsay2008} \cite{MFPCA_1}), \cite{escabias2004principal} and partial least squares (PLS) \cite{pls_exp} \cite{preda2007PLS}.  
\par In this paper we are focused on supervised classification with binary response $Y$ and multivariate functional data predictor $X  = (X^{(1)}, \ldots, X^{(d)})^{\T1}$, where for $j=1,\ldots,d$, $X^{(j)}$ are univariate functional random variables, $X^{(j)}=\{X^{(j)}(t), t\in \mathcal{I}_j\}$, and  $\mathcal{I}_j$ is some compact continuous index set. 
\par 
The supervised classification of univariate functional data ($d=1$)
has been the source of various contributions. \cite{FLDA} extended multivariate linear discriminant analysis (LDA) to 
irregularly sampled curves. As maximizing the between-class variance with respect to the total variance leads to an ill-posed problem, \cite{preda2007PLS} proposed a partial least square-based classification approach for univariate functional data. Using the concept of depth, \cite{lopez2006depth} introduced  robust procedures to classify 
functional data. Non-parametric approaches have also been investigated, using distances and similarities measures, see e.g \cite{ferraty}, and \cite{galeano2015} {\color{black} for an overview of the use of Mahalanobis distance}. Tree-based techniques applied to functional data classification are quite recent: \cite{maturo2022} introduced tree models using functional principal component scores as features, and \cite{moller2018} presented a tree based on curve distances.
\par  In the multivariate functional data setting, the supervised classification 
is mainly investigated when all domains $\mathcal{I}_j$ are identical,  $\mathcal{I}_j=[0, T]$ for $j=1, \ldots, d$ and $ T> 0$,  that is, all the $d$-components of $X$ are defined on the same domain. Under this assumption, \cite{blanquero2019} proposed a methodology that allows for optimal selection of the most informative time instants in the data. In \cite{gorecki}, regression models are used to classify multivariate functional data by reduction dimension techniques based on basis projection. Recently, \cite{LDA} proposed a  linear discriminant analysis. To avoid the ill-posed problem of the maximization  of the between variance in the functional case, the authors use discretization techniques, by pooling data at  specific time points. \par 
The classification of  multivariate functional data with different domains (the domains $\mathcal{I}_j$ are different) is rarely explored. This framework  is more flexible and  makes possible the use of different types of data simultaneously (e.g. time series, images) \citep{happ2018multivariate}. To the best of our knowledge, only \cite{golovkine} proposed a supervised classification method in this setting. They introduced a tree-based method for unsupervised clustering  and demonstrated the applicability of their method to supervised classification. Their method is based on principal component analysis for multivariate functional defined on different domains (MFPCA), presented in \citep{happ2018multivariate}.\par The use of MFPCA as ordinary principal component analysis (PCA) for supervised learning leads to some non-trivial issues, such as the
number and the selection of the principal components to be retained in the model. Therefore, the partial least square (PLS) approach has been an interesting alternative, as the obtained PLS components are based on  the relationship between predictors and the response. Since the introduction of PLS regression on univariate functional data predictors in \cite{preda2002regression}, many contributions have been proposed,  particularly in the univariate functional framework. As already mentioned above,  
\cite{preda2007PLS} demonstrated the ability to use PLS for linear discriminant analysis. In \cite{pls_exp} the authors show the relationship between the PLS of univariate functional and ordinary PLS on the coefficients obtained from basis expansion approximation.  An alternative non-iterative functional partial least square for regression is developed in \cite{delaigle2012}, and they  demonstrate consistency and establish convergence rates. For interpretability purposes, \cite{sparsefpls} has recently introduced a modified partial least squares approach to obtaining the sparsity of the coefficient function.\\ To the best of our knowledge,  PLS regression  for multivariate functional data has been explored only in one domain setting. In \cite{FPLS0}, the authors proposed a two-step approach for dealing with multivariate functional covariates. The first step consists of computing independently the PLS components for each (univariate) dimension $X^{(j)}, j=1, \ldots, d$. Then, they extract new uncorrelated features based on linear combinations of the obtained PLS components.  In \cite{MFPLS2022} the aim is to provide  a robust version of PLS for multivariate functional data. They have extended \cite{pls_exp} basis expansion  results on multivariate functional data and proposed a partial robust M-regression in this framework. 
\par We extend the recent contribution in \cite{MFPLS2022}  by investigating more exhaustively PLS procedures, in particular, we derive the relationship between the PLS  regression with univariate functional data (FPLS) and the PLS {\color{black} regression} with multivariate functional data (MFPLS).
{ \color{black} This relationship provides a way to estimate the PLS components for multivariate functional data  from the corresponding univariate ones. In the one-domain setting, it provides, from a computational point of view, a new estimating procedure. In the case of different domains, this relationship makes it possible to treat each univariate functional data with a different domain separately and then combine the FPLS components to obtain the MFPLS ones.
}

\par
The different dimension framework makes it possible to mix functional data of heterogeneous types (e.g. images and time series).  Inspired by  the Tree Penalized Linear Discriminant Analysis (TPLDA)  introduced in  \cite{poterie2019classification}, we propose a tree classifier based on PLS regression scores. Similarly to the TPLDA, our tree model uses a predefined structure group of dimensions. \par 
The paper is organized as follows. 
Section \ref{s_2} presents the PLS methodology for binary classification. It introduces the PLS regression with multivariate functional data defined on different domains and establishes the relationship with the univariate functional PLS approach. The presentation of the TMFPLS methodology ends this section. 
Section \ref{sec_sims} presents simulation studies for regression and classification purposes and compares the performances of our approaches with existing methods.  We also apply the MFPLS and TMFPLS methods to benchmark data for multivariate time series classification in Section \ref{appli}. A discussion is given in Section \ref{conc}. The appendix contains detailed
proofs of some theoretical results. The supplementary material includes  additional figures related to the numerical experiments.


\section{Methods}
\label{s_2}
\subsection{Basic principles and notations}
We are dealing with  multivariate functional data defined on different domains in a similar framework as \cite{happ2018multivariate}. As a general model for multivariate functional data analysis, let $X$ be a stochastic process represented by a \textit{d-}dimensional vector of functional random variables 
$X  =(X^{(1)}, \ldots, X^{(d) } )^{\T1}, $
defined on the probability space $(\Omega, \mathcal{A}, \mathbb{P})$. \par 
In the classical setting (\cite{ramsay2008}, \cite{MFPCA_1}, \cite{gorecki}), the components $X^{(j)}$, $j=1, \ldots, d$, are assumed real-valued stochastic processes defined on some finite continuous interval $[0,T]$. In our setting, we consider the general framework where each component $X^{(j)}$ is defined on some specific continuous compact domain $\mathcal{I}_j$ of $\mathbb{R}^{d_j}$, with $d_j \in \mathbb{N}-\{0\}$. Thus, for $d_j = 1$ we deal in general with time or wavelength domains whereas for $d_j=2$, the domain $\mathcal{I}_j$ indexes images or more complex shapes. It is also assumed that $X^{(j)}$ is a $L_2$-continuous process, and it has squared integrable paths, i.e. each trajectory of $X^{(j)}$ belongs to the Hilbert space of the square-integrable functions defined on $\mathcal{I}_j$,  $L_2(\mathcal{I}_j)$. These general hypotheses ensure that integrals involving the variables $X^{(j)}$ are well-defined. Let define  $\mathcal{H}= L_2(\mathcal{I}_1) \times ... \times L_2(\mathcal{I}_d)$ be the Hilbert space of vector functions

$$\mathcal{H} = \{f = (f^{(1)}, \ldots, f^{(d)})^{\T1}, f^{(j)} \in  L_{2}(\mathcal{I}_j)\;,  j=1, \ldots, d\}$$
endowed with the inner product 
 \begin{equation*}
     \langle \langle f, g\rangle \rangle_\mathcal{H} = \sum_{j=1}^d \langle f^{(j)}, g^{(j)}\rangle_{L_2(\mathcal{I}_j) } = \sum_{j=1}^d  \int_{\mathcal{I}_j} f^{(j)}(t)g^{(j)}(t) dt.
 \end{equation*}
 where $dt$ is the Lebesgue measure on $\mathcal{I}_j$. In the following, if there's no confusion, the index $\mathcal{H}$ will be omitted and  $\vert \vert \vert, \vert \vert \vert $ will denote the norm induced by $\langle \langle, \rangle \rangle$.
\subsection{The Linear Functional Regression Model}
When the aim is the prediction (the supervised context), the stochastic process $X$ is associated to a response variable of interest $Y$ through the conditional expectation $\mathbb{E}(Y\vert X).$\\ 
\noindent Let consider the real-valued response variable $Y$ be defined on the same probability space as $X$, $$Y:\Omega \rightarrow \mathbb{R}.$$
 Without loss of generality, we assume that $Y$ and $X$ are zero-mean, 
\begin{equation}
    \mathbb{E}(Y) = 0,\;  \mathbb{E}(X^{(j)}(t)) = 0 \;,  j =1, \ldots, d \;  t \in \mathcal{I}_j 
\label{z_mean}
\end{equation} 
and $Y$ has a finite variance.

The functional linear regression model  assumes that $\mathbb{E}(Y\vert X)$ exists and is a linear operator as a function of $X$. Thus, we have:  

\begin{equation}
    Y=\langle \langle X, \beta \rangle \rangle + \epsilon
    \label{eq_lm},
\end{equation}
where 
\begin{itemize}
\item[-] $\beta \in \mathcal{H}$ denotes the regression parameter (coefficient) function, 
$$\beta = \left(\beta^{(1)}, \ldots, \beta^{(d)}\right)^{\T1}, $$
\item[-] $\epsilon$ denotes the residual term which is assumed to be of finite variance $\mathbb{E}(\epsilon^2)= \sigma^2$ and uncorrelated to $X$.
\end{itemize}
In the integral form, the model in (\ref{eq_lm}) is written as :

\begin{equation}
    Y= \sum_{j=1}^{d}\int_{\mathcal{I}_j}X^{(j)}(t)\beta^{(j)}(t)dt + \epsilon.
    \label{eq_lm_int}
\end{equation}




%

\noindent Under the least squares criterion, the estimation of the coefficient function  $\beta$  is, in general, an ill-posed inverse problem (\cite{pls_exp}, \cite{preda2002regression}, \cite{preda2007PLS}). From a theoretical point of view, this is due to the infinite dimension of the predictor $X$, which makes that its covariance operator is not invertible \citep{cardot1999}. Hence, dimension reduction methods such as principal component analysis (\cite{happ2018multivariate}) and expansion of $X$ into a basis of functions (\cite{pls_exp}) can be used in order to obtain an approximation of linear form in (\ref{eq_lm_int}).  

\subsubsection{Expansion (of the predictor) into a basis of functions}
\label{exp}
For each dimension $j$, $j=1, \ldots, d$,  let consider in $L_2(\mathcal{I}_j)$ the set $\Psi^{(j)} = \{\psi^{(j)}_1, \ldots , \psi^{(j)}_{M_j}\}$ of $M_j$ linearly independent functions. Denote with $M = \sum_{j=1}^{d}M_j$.  

Assuming that the functional predictor $X$ and the regression coefficient function $\beta$ admit the expansions 
\begin{equation}
    X^{(j)}(t)  = \sum_{k=1}^{M_j} a_{k}^{(j)} \psi^{(j)}_k(t),  \ \ \ \; 
    \beta^{(j)}(t)  = \sum_{k=1}^{M_j} b^{(j)}_k \psi^{(j)}_k(t), 
    \label{b_1}
\end{equation}
 $\forall t \in \mathcal{I}_j, \; j = 1,\ldots, d$, the functional regression model in (\ref{eq_lm_int}) is equivalent to the multiple linear regression model: 

\begin{equation}
    Y= (\mathrm{F}a)^{\T1}  b + \epsilon 
    \label{reg_lin} 
\end{equation}
where
\begin{itemize}
\item[-] $a$ is the vector of size $M$ obtained by concatenation of vectors $a^{(j)}= (a^{(j)}_1, a^{(j)}_2, \ldots , a^{(j)}_{M_j} )^{\T1}$, $j = 1, \ldots, d, $
\item[-] $b$ is the coefficient vector of size $M$ obtained by concatenation of vectors  
$b^{(j)}= (b^{(j)}_1, b^{(j)}_2, \ldots , b^{(j)}_{M_j} )^{\T1}$, $j = 1, \ldots, d, $ and 
\item[-] $\mathrm{F}$ is the block matrix of size $M\times M$ with diagonal blocks $\mathrm{F}^{(j)}$, $j = 1, \ldots, d$, 

\begin{equation}
    \mathrm{F}= \begin{pmatrix} \mathrm{F}^{(1)} & 0 & \ldots& 0 \\ 
    0& \mathrm{F}^{(2)}& \ldots& 0 \\
    \vdots  &  & & \vdots   \\ 
    0 & 0 & \ldots & \mathrm{F}^{(d)}
    \label{F}
    \end{pmatrix}. 
\end{equation} 
For each $j = 1,\ldots, d$, $\mathrm{F}^{(j)}$ is the matrix of inner products between the basis functions, with elements $\mathrm{F}^{(j)}_{k,l} = \langle \psi_{k}^{(j)}, \psi_{l}^{(j)}\rangle_{L_2(\mathcal{I}_j)}$, $1\leq k,l\leq M_j$.

\end{itemize}

 \noindent Hence, under the assumption of basis expansion hypothesis (\ref{b_1}), the estimation of the  coefficient function, $\beta$, is equivalent to the estimation of the coefficient vector  $b$ in a classical multiple linear regression model with a design matrix involving the basis expansion coefficients of the predictor (the vector $a$) and the metric provided by the choice of the base's functions (the matrix $F$). \\ 
The least-square criterion for the estimation of $b$ yields in some settings (e.g. large number of basis functions) to  multicollinearity and high dimension issues, similar to the univariate setting (see \cite{pls_exp} for more details).  Two well-established methods of estimation,  principal component regression (PCR) and partial least squares regression (PLS) are reputed for the efficiency of their estimation algorithm and the interpretability of the results. As mentioned in \cite{jong1993pls} in the finite-dimensional setting and in \cite{pls_exp} for the functional one, for a fixed number of components, the PLS regression fits closer than the PCR. Thus, the PLS regression provides a more efficient solution (sum of square errors criterion). Numerical experiments confirm these results for the regression with univariate functional data (see for more details \cite{delaigle2012}, \cite{sparsefpls}). \\ 
In the next section, we present the proposed PLS regression of multivariate functional data. 
\par 
\subsection{PLS regression with multivariate functional data: MFPLS}
\label{PLS_m}
PLS regression penalizes the least squares criterion by maximizing the covariance between 
linear combinations of the predictor variables $X$ (the PLS components) and the response $Y$.
It is based on an iterative algorithm building at each step PLS components as predictors for the final regression model. 
In the multivariate setting, analogously to the univariate case (\cite{preda2002regression}), the weights for the linear combinations are obtained as the solution to the Tucker criterion:  

\begin{equation}
    \max_{w \in \mathcal{H} }\text{Cov}^2(\langle \langle w, X \rangle \rangle, Y) ,
    \label{f_1}
\end{equation}
with  $w = (w^{(1)}, \ldots, w^{(d)})^{\T1}$ such that $\vert\vert\vert w\vert\vert\vert_{\mathcal{H}}=1. $

The following proposition establishes the solution to the above maximization problem.
\begin{proposition} The solution of \eqref{f_1} is given by 
\begin{equation}
    w^{(j)}(t)=\frac{\mathbb{E}(X^{(j)}(t)Y)}{\sqrt{\sum_{k=1}^d\int_{\mathcal{I}_k}\mathbb{E}^2(X^{(k)}(s)Y)ds}}, \forall t\in \mathcal{I}_j,\; j=1, \ldots, d.
\end{equation}
\label{prop_mul}
\qed
\end{proposition}

\noindent Let denote by $\xi$  the PLS component defined as the linear combination of variables $X$ given by the weights $w$, i.e., 

\begin{equation*}
    \xi= \langle \langle X, w \rangle \rangle = 
    \sum_{k=1}^{d}\int_{\mathcal{I}_k}X^{(k)}(t)w^{(k)}(t)dt.
\end{equation*}

The iterative PLS algorithm works as follows:

\begin{itemize}
\item Step 0: Let $X_{0}=X$ and $Y_{0}=Y$. 
\item Step $h$, $h\geq 1$:  Define $w_h$ as in Proposition \ref{prop_mul} with $X = X_{h-1}$ and $Y= Y_{h-1}$. Then, define the $h$-th PLS component as

\begin{equation*}
    \xi_h = \langle \langle X_{h-1}, w_h \rangle \rangle \text{, }
\end{equation*}

Compute the residuals $X_h$ and $Y_h$ of  the linear regression of $X_{h-1}$ and $Y_{h-1}$ on $\xi_h$,

$$\begin{array}{l}
X_h = X_{h-1} - \rho_{h}\xi_{h}, \\
Y_h = Y_{h-1} - c_{h}\xi_{h}, \\
\end{array}$$
where $\rho_h = \displaystyle\frac{\mathbb{E}(X_{h-1}\xi_{h})}{\mathbb{E}(\xi_h^{2})} \in \mathcal{H}$ and $c_h = \displaystyle\frac{\mathbb{E}(Y_{h-1}\xi_{h})}{\mathbb{E}(\xi_h^{2})} \in \mathbb{R}$.

\item Go to the next step ($h = h+1)$.
\end{itemize}
{\color{black} Moreover, the following proprieties, stated in the univariate setting (Proposition 3 in \cite{preda2002regression}), are still valid in the multivariate case.} Let $L(X)$ denotes the linear space spanned by $X$.
\begin{proposition}{For any $h \geq 1$}, 
$\{\xi_k \}_{k=1}^h$ forms an orthogonal system of $L(X)$ and the following expansion formula hold:
\begin{itemize}
    \item[] \qquad $Y= c_1 \xi_{ 1}+ c_{ 2} \xi_{ 2} +...+ c_{ h} \xi_{ h} + Y_{h},$
    \bigskip
    \item[] \qquad $X^{(j)}(t)= \rho_{1}^{(j)}(t) \xi_{ 1} + \rho_{2}^{(j)}(t) \xi_{ 2} +...+ \rho_{h}^{(j)}(t) \xi_{h} + X_{h}^{(j)}(t),$ \\ $\forall t\in \mathcal{I}_j,\; j=1, \ldots, d. $ 
\end{itemize}
\label{prop_m}
\qed
\end{proposition}

The right-hand side term of the expansion of $Y$ provides the PLS approximation of order $h$ of (\ref{eq_lm}), 
\begin{equation} \langle\langle X, \beta\rangle\rangle\approx c_1 \xi_{ 1}+ c_{ 2} \xi_{ 2} +...+ c_{ h} \xi_{ h}
\label{pls_approx}
\end{equation} and the residual part, 
$$\varepsilon \approx Y_{h}.$$
Nevertheless, the above properties don't furnish the direct relationhip between  $Y$ and $X$ as in \eqref{eq_lm}. In order to do that, let's write the components $\xi_h$ as a linear function of $X$, i.e., as a dot product in $\mathcal{H}$. 
\begin{lem}{\em Let $\{v_{ k}\}_{k=1}^h$, $v_k$ $\in \text{span}\{w_1,…,w_h\}$,  be defined by :
\begin{equation}
    v_{ h}^{(j)}(t) = w_{ h}^{(j)}(t) - \sum_{k=1}^{h-1} \langle \langle \rho_{ k} , w_ { h} \rangle \rangle v_{k}^{(j)} (t),\;  t\in \mathcal{I}_j \ j= 1, \ldots, d.
    \label{lemm_eq} 
\end{equation}
Then, $\{v_{ k}\}_{k=1}^h$ forms a linearly independent system in $\mathcal{H}$ and

\begin{equation*}
    \xi_{h} = \langle \langle v_{ h},  X \rangle \rangle=\sum_{j=1}^{d}\int_{\mathcal{I}_j}X^{(j)}(t)v_h^{(j)}(t)dt.  
\end{equation*}
\label{lemm_multi}\qed
}
\end{lem}

\noindent Thus, as for the principal component analysis $X$  \citep{happ2018multivariate}, the PLS regression computes the components as the dot product of $X$ and the functions $\{v_k\}_{k=1}^h$. These components are suitable  for regression, as they capture the maximum amount of information between $X$ and $Y$ according to Tucker's criterion. Obviously, by Lemma \ref{lemm_multi} and (\ref{pls_approx}) we can write:  
\begin{equation*}
    c_1 \xi_{ 1}+ c_{ 2} \xi_{ 2} +...+ c_{ h} \xi_{ h} = \langle \langle X, \beta_{h} \rangle \rangle
\end{equation*}
with 
\begin{equation}
\beta_{h}= \sum_{i=1}^{h} c_i v_i.
\label{betah}
\end{equation}
Thus, the PLS regression model obtained with $h$ components is given by 

$$ Y= \langle \langle X, \beta_{h} \rangle \rangle  + \varepsilon_h, $$
with $\varepsilon_h = Y_h$. 

\par
\begin{remark} For $h\geq 1$, let $\mathcal{V}_h=\begin{pmatrix}
v_1 & \ldots & v_h     
\end{pmatrix}$ and $\mathcal{W}_h=\begin{pmatrix}
w_1 & \ldots & w_h     
\end{pmatrix}$ be two row-vectors of $\mathcal{H}^h$.
From Lemma \ref{lemm_multi}, the relationship between $\mathcal{V}_h$  and $\mathcal{W}_h$ given in (\ref{lemm_eq}) can be written in matrix form as
\begin{equation}
    \mathcal{V}_h = \mathcal{W}_h- \mathrm{P}_{h} \mathcal{V}_h, 
    \label{eq_rmrk}
\end{equation}
where $\mathrm{P}_{h}$ is the ${h\times h}$ matrix, 
\begin{equation*}
  \mathrm{P}_{h} = \left(
    \begin{array}{ccccc c }
    0 & \langle \langle \rho_1, w_2 \rangle \rangle& \langle \langle \rho_1 ,w_3 \rangle \rangle  &\ldots  & \langle \langle \rho_1, w_{h-1} \rangle \rangle&\langle \langle \rho_1, w_h \rangle \rangle \\  
   0  &  0 & \langle \langle \rho_2 ,w_3 \rangle \rangle  & 
 \ldots&\langle \langle \rho_2, w_{h-1} \rangle \rangle &\langle \langle \rho_2, w_h \rangle \rangle \\ 
         \vdots &    &    & \ldots   &  &  \vdots\\ 
          0& 0 & 0&  \ldots  & 0 &\langle \langle \rho_{h-1}, w_h \rangle \rangle  \\ 
          0&  0       & 0      & \ldots & 0 & 0 \\ 
  \end{array}\right)
\end{equation*}

Let $\mathbb{I}_{h \times h}$ be the identity matrix of size $h \times h$. Since $\mathbb{I}_{h \times h} + \mathrm{P}_{h}$ is non-singular, 
then equation \eqref{eq_rmrk} yields to   
\begin{equation}
\mathcal{V}_h=(\mathbb{I}_{h \times h } + \mathrm{P}_{h} )^{-1}  \mathcal{W}_h.   
\label{vfuncw}
\end{equation}
As the functions $\{w_{k}\}_{k=1}^h$ are computed directly from Proposition \ref{prop_mul}, the equation (\ref{vfuncw}) provides a straightforward way to obtain the weights $\{v_k\}_{k=1}^h$ and, therefore by (\ref{betah}), we obtain the coefficient function approximation $\beta_h$.
\end{remark}

\begin{remark}
It is worth noting that, contrary to eigenfunctions of the PCA of $X$ \cite{happ2018multivariate}, $\{ v_k \}_{k=1}^h$ is not an orthogonal system by the inner product $\langle \langle,\rangle  \rangle $. Nonetheless, it  provides orthogonal PLS components, i.e. $\mathbb{E}(\xi_{k} \xi_{l} ) = \mathbb{E}(\xi_{k}^2) \delta_{k,l}$, where $\delta_{k,l} $ is the  Kronecker symbol. 

\end{remark}

The next part focuses on the relationship between the partial least squares regression of  univariate functional data (FPLS) and the proposed multivariate version (MFPLS) given by Proposition \ref{prop_mul}. More precisely, we show that the MFPLS regression can be solved by iterating FPLS within a two-stage approach.  
\subsubsection{Relationship between MFPLS and FPLS}

For each $j$ in $1, \ldots d$, let $\tilde{w}_1^{(j)}$ be the weight function corresponding to the first PLS component obtained by FPLS regression of $Y$ on dimension $X^{(j)}$ (Proposition 2 in \cite{preda2002regression}), 

\begin{equation*}
    \tilde{w}_1^{(j)}(t)= \frac{\mathbb{E}(X^{(j)}(t)Y) }{ \sqrt{\int_{\mathcal{I}_j} \mathbb{E}^2(X^{(j)}(s)Y)ds} }, \; t \in \mathcal{I}_j.
\end{equation*}
Obviously,  the functions $\tilde{w}^{(j)}_1$ and $w^{(j)}_1$ (as defined in Proposition \ref{prop_mul}) are related by the following relationship

\begin{equation}
    w^{(j)}_1=u_{1, j} \tilde{w}^{(j)}_1, \
 \end{equation}   
with       
$u_{1, j} = \displaystyle \frac{  || \mathbb{E}(X^{(j)}Y) ||_{L_2(\mathcal{I}_j)}  }{|||\mathbb{E}(XY)|||} \in \mathbb{R}$
and  $|| \cdot||_{L_2(\mathcal{I}_j) }$ stands for the norm induced by the usual inner product in $L_2(\mathcal{I}_j) $.\\

\noindent Let note that the vector    
$ u_1=\begin{pmatrix} u_{1,1}, ..., u_{1,d}\end{pmatrix}^{\T1}
    $ is such that $ \vert \vert u\vert \vert _{\mathbb{R}^{d} }=1$. Hence, we can establish a relationship between the PLS components of MFPLS and FPLS in the following way. 
    {\it For each $j$ in $1, \ldots, d$,  let denote by $\xi^{(j)}_1$ the first PLS component obtained by the FPLS regression of $Y$ on the $j$-th dimension of $X$. 
    Then, the first PLS component of the MFPLS of $Y$ on $X$, $\xi_1$, is obtained as the first PLS component of the PLS regression of $Y$ on $\{\xi_{1}^{(\color{black}1)}, \ldots, \xi_{1}^{(d)}\}$. 
}
    \noindent Based on the iterative PLS process, that relationship can be applied to the computation of higher-order MFPLS components.  
\par  That relationship allows us to define a new methodology for the computation of MFPLS regression when the functional predictor $X$ is approximated into a basis of functions. 

\subsubsection{Using the basis expansion for the MFPLS algorithm}
\label{mfPLS}

Under the hypothesis \eqref{b_1}, for any $j = 1, \ldots d$, let denote with $\Lambda^{(j)}$ the vector 
$$\Lambda^{(j)}= \left(\mathrm{F}^{(j)}\right)^{1/2}a^{(j)} $$
where $\left(\mathrm{F}^{(j)}\right)^{1/2}$ is the squared root of the matrix
$\mathrm{F}^{(j)}$ and $a^{(j)}$ is the vector of (random) coefficients of the expansion of $X^{(j)}$ in the basis of functions $\{\psi^{(j)}_1, \ldots, \psi^{(j)}_{M_j}\}$. 

Then, Proposition 2 in \cite{pls_exp} makes the MFPLS procedure equivalent to the following algorithm.  

\paragraph{\bf MFPLS algorithm}

\begin{itemize}
\item Step 0: Let $\Lambda_{0}^{(j)}=\Lambda^{(j)}$ for all $j =1, \ldots, d$  and $Y_{0}=Y$. 
\item Step $h$, $h\geq 1$:  
\begin{enumerate}
    \item For each $j =1, \ldots, d$,
    \begin{itemize}
    \item[-]define $\xi^{(j)}_h$ as the {\bf first} PLS component obtained by the ordinary PLS regression of $Y_{h-1}$ on $\Lambda_{h-1}^{(j)}$, 
\begin{equation}
    \xi^{(j)}_{h} = \sum_{k=1}^{M_j}\Lambda^{(j)}_{h-1, k} \theta_{h,k}^{(j)},   
\end{equation}
where $\theta_h^{(j)} \in \mathbb{R}^{M_j}$ is the associated weight vector. 
\end{itemize}

\item Define the $h$-th MFPLS component $\xi_h$ by  as the {\bf first} PLS component of the regression of $Y_h$ on  $\{\xi^{(1)}_h, \ldots, \xi^{(d)}_h\}$, 
 \begin{equation}
     \xi_h=\sum_{k=1}^{d}\xi^{(k)}_{h}u_{h,k},  
 \end{equation}
where $u_{h}\in \mathbb{R}^{d}$ is the associated weight vector. 

\item 
\begin{itemize}
\item For each $j=1, \ldots, d$, compute the residuals $\Lambda^{(j)}_h$ of the linear regression of $\Lambda^{(j)}_{h-1}$ and $Y_{h-1}$ on $\xi_h$,
$$
\Lambda^{(j)}_h = \Lambda^{(j)}_{h-1} - r^{(j)}_h{\xi_{h}}, 
$$
where $
    r^{(j)}_h=\displaystyle\frac{\mathbb{E}(\xi_h\Lambda^{(j)}_{h-1})}{\mathbb{E}(\xi_h^2) } \in \mathbb{R}^{M_j} 
$. 
\item compute the residual $Y_h$ of the linear regression of $Y_{h-1}$ on  $\xi_h$, 

$$Y_h = Y_{h-1} - c_{h}\xi_{h}, \\
$$
where 
$c_h = \displaystyle\frac{\mathbb{E}(Y_{h-1}\xi_{h})}{\mathbb{E}(\xi_h^{2})} \in \mathbb{R}$.
\end{itemize}
\end{enumerate}
\item Go to the next step ($h = h+1)$.
\end{itemize}
\begin{remark} 
\ 
\begin{itemize}
\small 
\item[1.] The number of PLS components ($h$) retained in the approximation of the regression model (\ref{pls_approx}) is usually chosen by cross-validation optimizing some criteria {\color{black}as the Mean Squared Error (MSE) or, for binary classification, the Area under the ROC Curve (AUC)}.
\item[2.] The approach of \cite{MFPLS2022} is an extension of the  basis expansion result from \cite{pls_exp}. It was proposed for one domain definition. Note that our approach is more flexible since it allows different intervals. The case of one domain is then a special case of the proposed methodology (see Section \ref{reg_one} for numerical comparison).
\item[3.] Let introduce in the step $h$ - point 3 of the algorithm the computation of functions $w_h^{(j)}$ and $\rho_h^{(j}$, $j = 1,\ldots, d$, as
\begin{align*}
     w_h^{(j)}&= u_{h, j} \mathrm{H}^{(j)}\theta_{h}^{(j)} \psi^{(j)},  \\
     \rho_h^{(j)}&=\mathrm{H}^{(j)}r_h^{(j)}\psi^{(j)} 
\end{align*}
where $ \mathrm{H}^{(j)}= \left(\mathrm{F}^{(j)}\right)^{-1/2}$. \\ Then, by Lemma \ref{lemm_multi}, the functions $\{v_k\}_{k=1}^{h}$ can also be computed through the MFPLS algorithm,  allowing to compute the regression coefficient function in (\ref{betah}).
\item[4.]\textbf{Computational details:} Let consider $(\mathcal{X}_1, \mathcal{Y}_1), \ldots, \mathcal{(X}_n,\mathcal{Y}_n)$ be an  i.i.d. sample of size $n \geq 1$ of $(X,Y)$. Then, for each $j = 1, \ldots, d$, the vector $a^{(j)}$ is represented by the sample $n\times M_j$ matrix of coefficients $\mathbf{A}^{(j)}$, 
 \begin{equation*}
    \mathbf{A}^{(j)} = \begin{pmatrix} a^{(j)}_{1,1} & \ldots & a^{(j)}_{1,l} & \ldots &a^{(j)}_{1,M_j}\\
    \vdots & & \vdots & & \vdots \\
     a^{(j)}_{k,1} & \ldots & a^{(j)}_{k,l} & \ldots &a^{(j)}_{k,M_j}\\
      \vdots &  &  \vdots & & \vdots \\
       a^{(j)}_{n,1} & \ldots & a^{(j)}_{n,l} & \ldots &a^{(j)}_{n,M_j}.
    \end{pmatrix} 
\end{equation*}
and 
$$ \mathbf{\Lambda}^{(j)}=  \mathbf{A}^{(j)}(\mathrm{F}^{(j)})^{1/2}. 
$$
 Let define $\mathbf{Y} =(\mathcal{Y}_1, \ldots, 
 \mathcal{Y}_n)^{\T1}$. Then, the matrix version of the  MFPLS algorithm (step $h$) can be rewritten as :
\begin{enumerate}
    \item[1] For each $j=1, \ldots, d$, define $\boldsymbol{\xi}_h^{(j)}\in \mathbb{R}^n$ as the first PLS component obtained by the ordinary PLS of $\mathbf{Y}_h$ on $\mathbf{\Lambda}_{h-1}^{(j)}$
    \item[2] The $h$-th MFPLS $\boldsymbol{\xi}_h$ is the first component obtained by the ordinary PLS of $\mathbf{Y}_h$  on  $\left(\boldsymbol{\xi}^{(1)}_h, \ldots, \boldsymbol{\xi}^{(d)}_h\right)$. \\ 
    \item[3] For $j=1, \ldots, d$, the residuals $\boldsymbol{\Lambda}_h^{(j)}$ are computed by   
    \begin{align*}
        \mathbf{\Lambda}_h^{(j)}=   \mathbf{\Lambda}_{h-1}^{(j)}- \boldsymbol{\xi}_h \mathbf{r}^{(j)}_h
    \end{align*}
    with $\mathbf{r}_h^{(j)}=\displaystyle\frac{1}{\boldsymbol{\xi}_h^{\T1} \boldsymbol{\xi}_h}\boldsymbol{\xi}_h^{\T1}\mathbf{\Lambda}^{(j)}_{h-1} $ is the projection coefficient. \\ 
   And the residual $\mathbf{Y}_{h}$ is  
    \begin{align*}
        \mathbf{Y}_h= \mathbf{Y}_{h-1}- \boldsymbol{\xi}_h\mathbf{c}_h
    \end{align*}
    with $\mathbf{c}_h=\displaystyle\frac{1}{\boldsymbol{\xi}_h^{\T1} \boldsymbol{\xi}_h}\boldsymbol{\xi}_h^{\T1}     \mathbf{Y}_{h-1} $. 
\end{enumerate}
\end{itemize}

\end{remark}

{\color{black}
Although the proposed methodology is for regression problems with scalar response, it can be used for binary classification by using the relationship between
linear discriminant analysis and linear regression  (\cite{pls_exp}, \cite{preda2007PLS}). The next section addresses a classification application based on PLS regression.}

\subsubsection{From PLS regression to PLS binary-classification}
\label{reg_class}
Using the previous notations, let $X$ be the predictor variable (not necessarily zero-mean) and $Y$ be the response.  The binary classification setting assumes that $Y$ is a Bernoulli variable,  $Y \in \{0, 1\}$, $Y\sim \mathcal{B}(\pi_1)$  with $\pi_1= \mathbb{P}(Y=1)$. The PLS regression can be extended to binary classification after a convenient encoding of the response.\\ Let define the variable $Y^*$ as 

\begin{equation}
Y^* = \left\{
    \begin{array}{ll}
     \sqrt{\frac{\pi_1}{\pi_0}}     &\text{, if }  Y=0\\
        - \sqrt{\frac{\pi_0}{\pi_1}} &\text{, if } Y=1, 
    \end{array}
\right.
\label{reco}
\end{equation}
with $\pi_0 =1-\pi_1$.\\ 
 Then, the coefficient function $\beta$ of the regression of $Y^{*}$ on $X$ corresponds (up to a constant) to that defining the 
 the Fisher discriminant score denoted by $\Gamma(X)$: 
 
\begin{equation*}
    \Gamma(X)= \alpha + \langle \langle X,\beta\rangle \rangle,  
\end{equation*}
with $\alpha$= -$\langle \langle \mu , \beta \rangle \rangle $, and $\mu = \mathbb{E}(X) \in \mathcal{H}$.\\  
Finally, the predicted class  $\hat{Y}_0$ of a new curve  $X_{0}$ is given by 
\begin{equation*}
    \hat{Y}_0 = \left\{
    \begin{array}{ll}
         0&  \text{if }  \Gamma(X_{0}) > 0\\
        1 & \text{otherwise. }
    \end{array}
\right. 
\end{equation*}
See for more details \cite{preda2007PLS}. \\  
In this paper we estimate the coefficient function $\beta$ by the MFPLS approach.
\subsection{MFPLS tree-based methods}
\label{sec_T}
{\color{black}
Alternatives to linear models such {\color{black} Support Vector Machine or SVM} (see e.g \cite{svm} \cite{blanquero2019_2}), clusterwise regression (see e.g \cite{PLS_c}, \cite{yao2011}, \cite{LI2021}) could be 
extended to multivariate functional data. In this section, we present a tree-based methodology (TMFPLS) combined with MFPLS models. The variable selection feature of tree methods is particularly adapted in the framework of multivariate functional data and allows predicting the response throughout more complex but still interpretable relationships.  It represents in some way a generalization of the finite-dimensional setting presented in \cite{poterie2019classification} to the case of multivariate functional data. The procedure consists in split a node of the tree by successively selecting an optimal discriminant score (according to some impurity measure) among discriminant scores obtained from MFPLS regression models with different subsets of predictors. In the presented methodology, we limit our attention  to the case of binary classification ($Y \in \{ 0, 1\}$ ).
\subsubsection{The algorithm}
 Let consider $(\mathcal{X}_1, \mathcal{Y}_1), \ldots, \mathcal{(X}_n,\mathcal{Y}_n)$ be an  i.i.d. sample of size $n \geq 1$ of $(X,Y)$, where $Y$ is a binary response  variable and $X = (X^{(1)}, \ldots, X^{(d)})$  a multivariate functional one. Moreover, we assume that there exists a well-defined group structure (potentially overlapping) of the dimensions of $X$, i.e. there exists $K$ subgroups, $K\geq 1$,  $\mathcal{G}_1, \ldots, \mathcal{G}_K$ of variables $\{X^{(1)}, \ldots, X^{(d)}\}$.  Notice that groups are not necessarily disjoint. These groups of variables define the candidates to be used to split the node of the tree (the score will be calculated with the only variables in the candidate group). 
 
    
Inspired by  \cite{poterie2019classification}'s methodology, our algorithm is composed of two main steps. In a nutshell, with the help of MFPLS methodology, the first step provides the results of the splitting according to candidates groups $\mathcal{G}_1, \ldots, \mathcal{G}_K$ whereas the second one selects the best splitting candidate using an impurity criterion. These two steps are applied to all current nodes (start with the root node containing all the sample -- $n$ observations) until the minimum purity threshold is reached. \\ \\  
Consider the current node of the tree to be split:
\begin{itemize}
    \item \textbf{Step 1: The MFPLS candidate scores}. \\ For each candidate group of variable $\mathcal{G}_{i}$, $i  = 1, \ldots, K$,  perform MFPLS of $Y$ on $\mathcal{G}_{i}$ and denote with $\Gamma^{i}$ the estimated MFPLS score (prediction) obtained with the  group $\mathcal{G}_{i}$. Then, the result of the split with $\Gamma^{i}$ is represented by two new sub-nodes obtained according to the predictions of the observations $(\mathcal{X}_j, \mathcal{Y}_j)$ in the current node: $\{\Gamma^{i} (\mathcal{X}) > 0  \}$ and  $\{\Gamma^{i} (\mathcal{X}) \leq 0 \}$.

\item \textbf{Step 2: Optimal splitting}.\\ Select the optimal splitting  according to group $\mathcal{G}^{*}$ which 
maximizes the decrease of \textit{impurity} function $\Delta \mathcal{Q}$ (see \cite{poterie2019classification} for more details),  
\begin{equation*}
    \mathcal{G}^*=\argmax_{\mathcal{G}\in \{\mathcal{G}_1, \ldots, \mathcal{G}_K \} }\Delta\mathcal{Q}^{\mathcal{G}}.
\end{equation*}

Therefore, the optimal splitting for the current node is the one obtained with the MFPLS score $\Gamma$ corresponding to $\mathcal{G}^*$.  
\end{itemize}
A node is terminal if its impurity index is lower than a defined purity threshold.
In order to avoid overfitting, a pruning method can be employed. Here, we use the same technique as in \cite{poterie2019classification}, i.e.  the optimal depth of the decision tree ($m^*$) is estimated using a validation set.

\section{Simulation study}
This section deals with finite sample properties on simulated data to evaluate the performances of MFPLS and TMFPLS approaches  with competitor methods based on MFPCA. Two different cases are presented. In the first one, all the components $X^{(j)}$ of $X$ are defined on the same one domain $\mathcal{I}=[0,T]$, $T > 0$. In the second case, $X$ is a bivariate functional vector $X = (X^{(1)}, X^{(2)} )^{\T1}$  with  $X^{(1)}=(X^{(1)}(t))_{t\in \mathcal{I}_1}$  and a two-domain variable $X^{(2)}=(X^{(2)}(t))_{t\in \mathcal{I}_2\times  \mathcal{I}_2}$, where $\mathcal{I}_1, \mathcal{I}_2 \subset \mathbb{R}$. Thus, in this second one, a sample from the functional variable $X$ corresponds to a set of curves  and 2-D images of domains $\mathcal{I}_1$ and $\mathcal{I}_2\times \mathcal{I}_2$ respectively. \par 
{\color{black}All computation results reported in this section were obtained using a computer that has a
Windows $10$ operating system, $11$th Gen Intel(R) Core(TM) i$7$-$1165$G$7$ $2.80$GHz  and $16.00$Go Go of RAM memory. }
\label{sec_sims}

\subsection{One domain case} 

\subsubsection{ {Setting 1:} scalar response}
\label{reg_one}
{ \color{black}  Since the main concurrent MFPLS was proposed for regression \citep{MFPLS2022}, we conduct in this section the regression simulation framework described in \cite{MFPLS2022} and compare their method with MFPLS.} \par  Consider the domain  $\mathcal{I}= [0,1]$ and the 3-dimensional functional predictor $X= (X^{(1)}, X^{(2)}, X^{(3)} )^{\T1} $:
\begin{equation*}
    X^{(j)}(t)=\sum_{k=1}^5 \gamma_k \upsilon_k(t),\;  \text{  } t\in \mathcal{I},\,  j= 1, 2, 3,    
\end{equation*}
with $\gamma_k \sim \mathcal{N}(0, 4k^{-3/2} ) $ and $\upsilon_k(t) = \sin{k\pi t}- \cos{k \pi t}$, $k=1, \ldots, 5$. \\ 
The functional  coefficient  $\beta$  is defined by  

\begin{equation*}
    \beta(t) =\begin{pmatrix}
    \sin(2\pi t), 
    \sin(3\pi t), 
    \cos(2\pi t)
    \end{pmatrix}^{\T1}. 
\end{equation*}
Then, the regression model generating the data is given by 
\begin{equation*}
    Y= \langle \langle X, \beta \rangle \rangle + \epsilon,
\end{equation*}
 where $\epsilon \sim \mathcal{N}(0, \sigma^2)$.\\ 
We define the noise variance as 
$$\sigma^2= \frac{\mathbb{E}(\langle \langle X, \beta \rangle\rangle^2) }{\text{SNR} },$$
where SNR is the signal-to-noise ratio. We consider  5 values of SNR: SNR $\in \{0.5, 1.62, 2.75, 3.88, 5\} $.\\   
The approach proposed in \cite{MFPLS2022} is a generalization of  the result in \cite{pls_exp}  to the multivariate case (MFPLS\_D). It exploits an equivalence between the PLS of multivariate functional covariates and ordinary PLS of the projection scores of covariates in basis functions.  Our method has a different procedure, as we compute multivariate  PLS components using the univariate PLS components (see Section \ref{mfPLS}).\\
\par As in \cite{MFPLS2022} we use $200$ equidistant discrete times points on $\mathcal{I}$ where raw data of $X$ are observed, and 400 independent copies of $X$ are simulated. Among these copies,  $50\%$ are used for learning and the remaining for validation.\\
We also compare our method to principal component regression\footnote{From \cite{MFPLS2022} scripts: \url{https://github.com/UfukBeyaztas/RFPLS}}(MFPCR). 
A number of $200$ replications of the three different inference procedures are done. \\ The number of components in all approaches is chosen by 10-fold cross-validation procedures. To transform the raw data into functions, smoothing is used with $20$ quadratic splines basis functions.\\ 
Performances of the three approaches are measured by the
mean squared prediction error (MSPE): 
\begin{equation*}
    \text{MSPE}= \frac{1}{200}\sum_{i\in V_{\text{set} }}(Y_i - \hat{Y}_i)^2
\end{equation*}
with $\hat{Y}_i$ is the predicted response for the $i-th$
observation in the validation sample ($V_\text{set}$), $Y_i$ the true value. 
\par {\color{black} The results of the experiences are summarized in Table \ref{comp}.}

\begin{table}[ht]
\centering
\color{black}
\begin{tabular}{c c c c c c c c c }
  \hline
 & \textbf{SNR }& \multicolumn{2}{c}{\textbf{MFPLS}} & \multicolumn{2}{c}{\textbf{MFPLS\_D}} & \multicolumn{2}{c}{\textbf{MFPCR}} \\ 
 &  & \textit{MSPE} & \textit{Time (in s)} &\textit{MSPE} & \textit{Time (in s)} &\textit{MSPE} & \textit{Time (in s)}  \\ 
 & 0.50 & 10.85(1.56) &  0.79(0.40) & 10.63(1.48) &0.13(0.05)  & 10.73(1.48)& 0.19(0.08) \\ 
   & 1.62 & 3.42(0.45) &  0.74(0.06)& 3.43(0.45)&0.11(0.05) & 3.45(0.45)&0.14(0.02) \\ 
   & 2.75 & 2.04(0.30) & 0.79(0.09) & 2.07(0.29) & 0.10(0.01) & 2.08(0.31) &0.15(0.02) \\ 
   & 3.88 & 1.45(0.26) & 0.75(0.07)& 1.47(0.22) & 0.13(0.08)& 1.47(0.22)&  0.15(0.06)\\ 
   & 5.00 & 1.11(0.16) &  0.73(0.06)& 1.15(0.17) & 0.10(0.04) & 1.15(0.18)& 0.14(0.03) \\ 
   \hline
\end{tabular}
\caption{ \color{black} Means and standard deviations (in parentheses) of MSPEs obtained and estimation times in the experiments. }
\label{comp}
\end{table}

        
All methods provide comparable results. {\color{black} Although the proposed method (MFPLS) is more time-consuming (around $60$ milliseconds), Table \ref{comp} shows that MFPLS and MFPLS\_D give similar performances. } 


\subsubsection{Setting 2: binary response} 
In the following experiment, we focus on a classification problem with one-dimensional domain $\mathcal{I}= [0,1]$.\\ Here, we build two different classes (Class 1 and Class 2) of functional data, visualized in Figure \ref{data_plot}. They are related to some pattern (a shape with peaks)  appearing at some locations of the curves. This simulation setting is a kind of visual pattern-recognition problem  \citep{visual}. The pattern of interest is identifiable by eyes but challenging to detect with algorithms. An example may be epileptic spikes detection in electroencephalogram recordings (see for instance \cite{spikes} for more details). \\
The performances of MFPLS, TMFPLS, and linear discriminant analysis on principal component scores (MFPCA-LDA) are compared in this simulation study.\\ \par 
{\color{black}
Consider the domain $\mathcal{I}=[0, 1]$, and the 2-dimensional functional predictor $X=(X^{(1)}, X^{(2)})^{\T1} $ : 
\begin{align*}
    X^{(1)}(t) = \sum_{s=1}^4 a_{s} h_s(t) + \epsilon^{(1)}(t)\;,   X^{(2)}(t)=\sum_{s=1}^4 (1- a_{s}) h_s(t) + \epsilon^{(2)}(t) ,
\end{align*}
 where 
 \begin{itemize}
     \item $a=\begin{pmatrix}
     a_1 & a_2& a_3& a_4
 \end{pmatrix}^\top$ is a random vector taking values in $\{ -1, 0, 1 \}^4$. Note that in this case, $a$ can possibly take $3^4=81$ different values, which are denoted by $V_1, \ldots, V_{81}$.
 \item $h_1, \ldots, h_4$ are triangle functions: $$h_s(t)= \left(1- 10|t- u_s|\right)_+, s=1, \ldots, 4,$$ 
 with $u_1=0.2$, $u_2=0.4$, $u_3=0.6$, and $u_4=0.8$.
 \item $\epsilon^{(1)}$ and $\epsilon^{(2)}$ are two independent white noises functions with variance ${Var}(\epsilon^{(j)}(t))=0.20$, for $j=1, 2$ and $t\in \mathcal{I}$.
 \end{itemize} 
\noindent We consider that $V_1=\begin{pmatrix}
     1 & 1& 0 & 0
 \end{pmatrix}^\top$, $V_2=\begin{pmatrix}
     0 & 1& 1 & 0
 \end{pmatrix}^\top$, $V_3=-V_1$ and $V_4=-V_2$. The vectors $V_5, \ldots, V_{81}$ are the $76$ other possible values of $a$. Let the response variable $Y$ be
\begin{equation*}
Y = \left\{
    \begin{array}{ll}
        1 & \mbox{if } a \in \{V_1, V_2, V_3, V_4 \}
         \textbf{ }  \\ 
        0 & \mbox{if } a \in \{V_5, \ldots, V_{81}\}
    \end{array}
\right.
\end{equation*} 
 If $a_s\neq 0$, $s=1, \ldots, 4$,  this means we observe a peak on $X^{(1)}$ at the position $t=u_s$, which could be positive ($a_s=1$) or negative ($a_s=-1$). Then, $Y=1$, if two consecutive peaks (both negatives or positives) occur at the beginning of $X^{(1)}$. Namely the peaks of interest are observed at $t=0.2, 0.4$ or $t=0.4, 0.6$. The four cases where $Y=1$ are illustrated in Figure \ref{s_evn}. \\

\begin{figure}[H]
\centering
\begin{tabular}{c c   c} 
 & $X^{(1)}$  & $X^{(2)}$ \\
    $a=V_1$ &  \includegraphics[scale=.12, align=c]{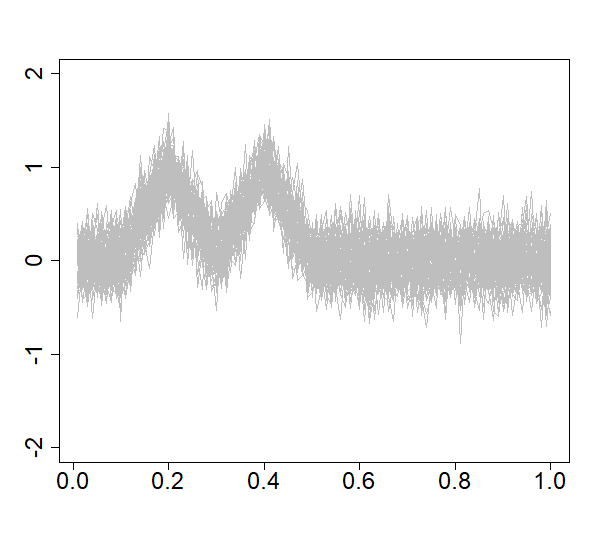} 
& \includegraphics[scale=.12, align=c]{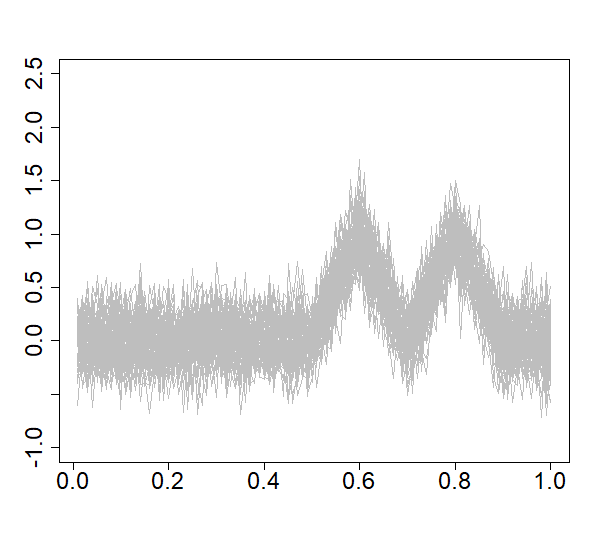}  \\ 
    $a=V_2$ &  \includegraphics[scale=.12, align=c]{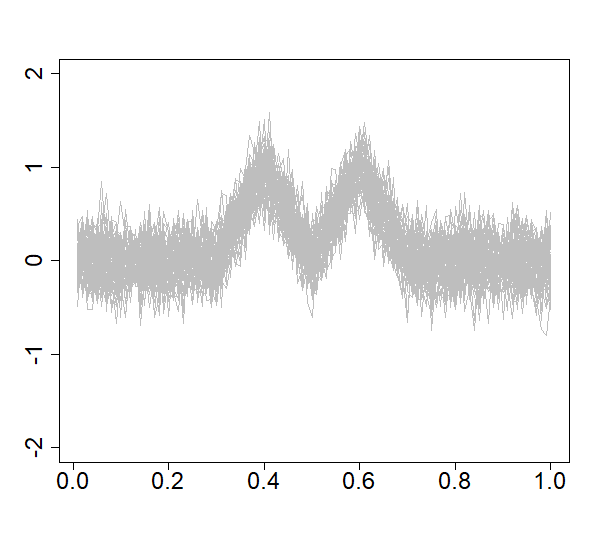} 
  & \includegraphics[scale=.12, align=c]{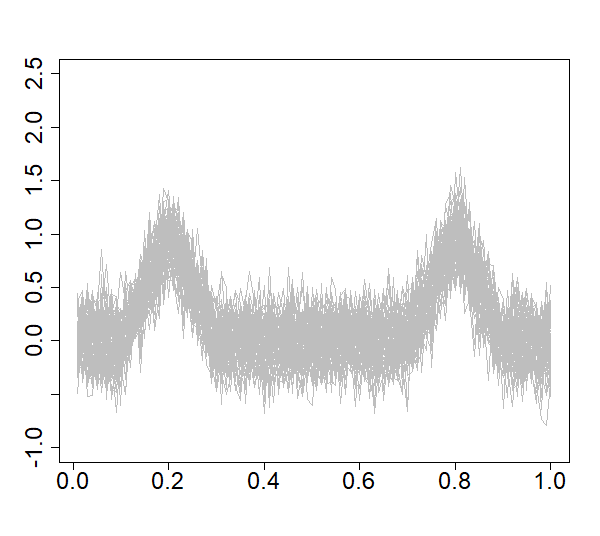}  \\ 
$a=V_3$ &   \includegraphics[scale=.12, align=c]{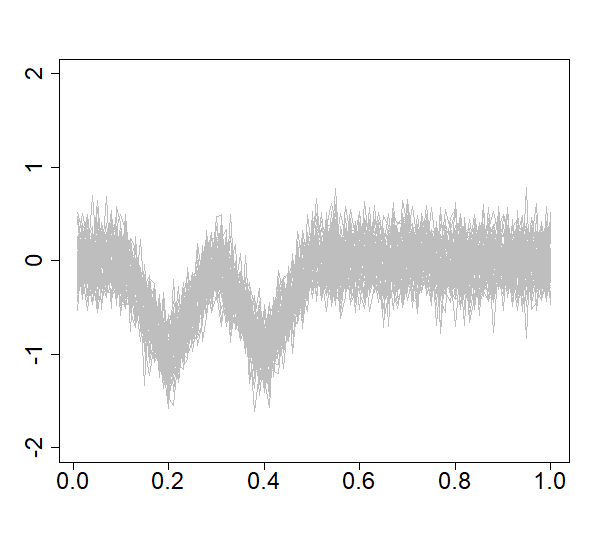} 
  & \includegraphics[scale=.12, align=c]{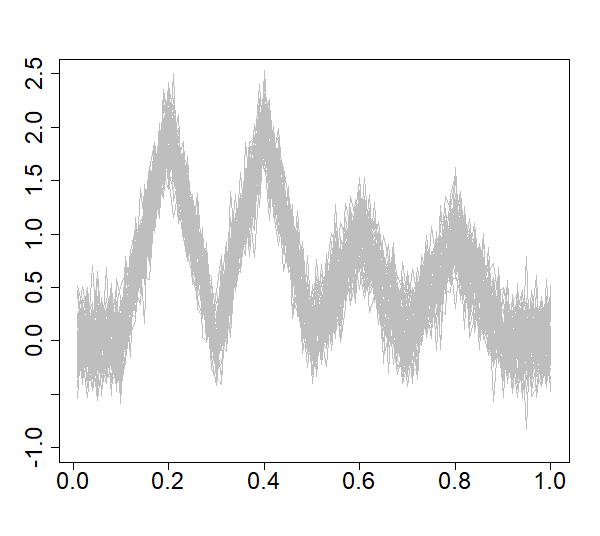}  \\ 
 $a=V_4$ &  \includegraphics[scale=.12, align=c]{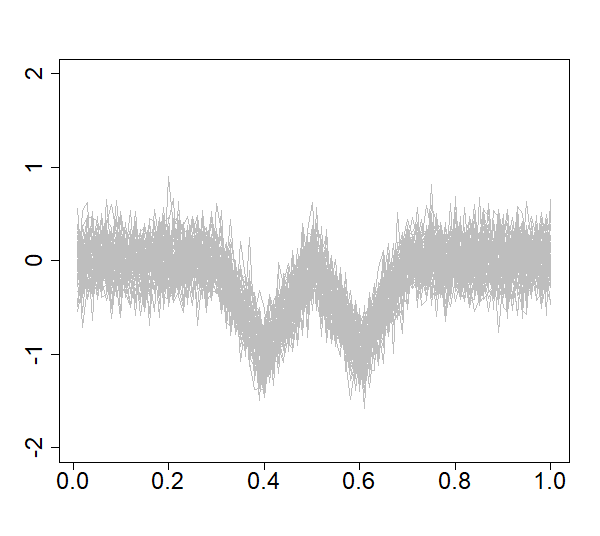} 
  &  \includegraphics[scale=.12, align=c]{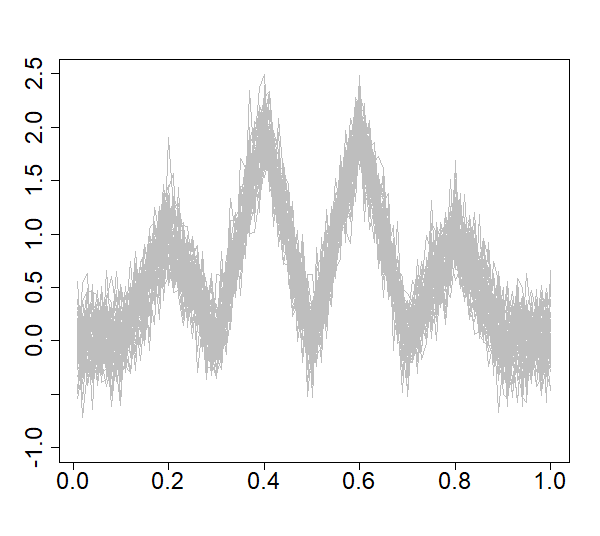} 
\end{tabular}
   \caption{Examples of $N=25$ realizations of $X$ for $Y=1$ (Setting 2)}. \\ 
       \label{s_evn}
\end{figure}
According to the distribution of $a$, we study two scenarios (see Figure \ref{data_plot}) : 
\begin{itemize}
    \item \textbf{Scenario 1:} Positive peaks. 
    \subitem $\mathbb{P}(a=V_1)=\mathbb{P}(a=V_2)=\frac{1}{4}$ and  $\mathbb{P}(a=V_3)=\mathbb{P}(a=V_4)=0$. 
    \subitem $\mathbb{P}(a=V_5)= \mathbb{P}(a=V_6)= \ldots= \mathbb{P}(a=V_{81})=\frac{1}{162}$.   
    \item \textbf{Scenario 2:} Positive and negative peaks. 
    \subitem $\mathbb{P}(a=V_1)=\mathbb{P}(a=V_2)=\mathbb{P}(a=V_3)= \mathbb{P}(a=V_4)=\frac{1}{8}$.
    \subitem $\mathbb{P}(a=V_5)= \mathbb{P}(a=V_6)= \ldots= \mathbb{P}(a=V_{81})=\frac{1}{162}$.
\end{itemize}

    

\begin{figure}[ht]
\centering
 \begin{tabular}{ c  c |  c  }
 & \textbf{Scenario 1} &  \textbf{Scenario 2} \\ 
$Y=0$   & \begin{tabular}{c c}
$X^{(1)}$ & $X^{(2)}$ \\  
     \includegraphics[align=c, scale=.1]{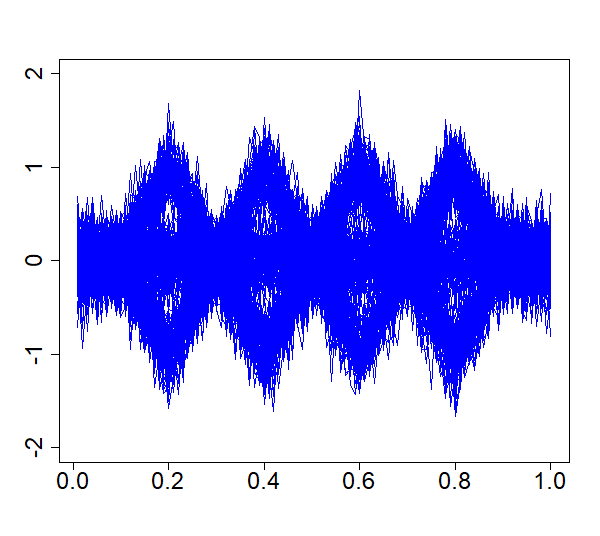} & \includegraphics[align=c, scale=.1]{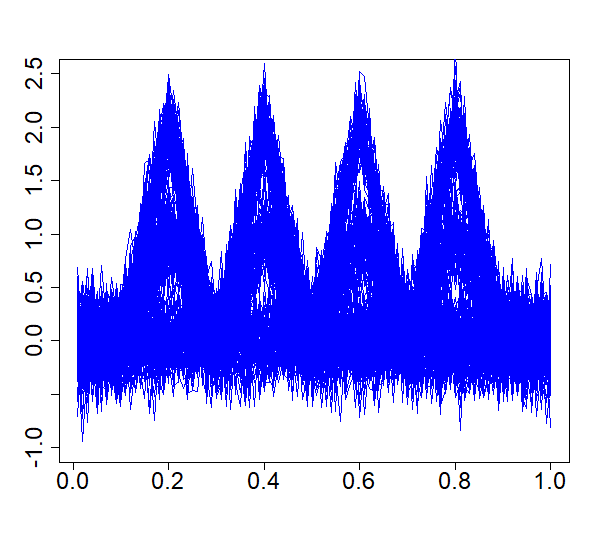}  
\end{tabular} & \begin{tabular}{c c}
$X^{(1)}$ & $X^{(2)}$ \\ 
     \includegraphics[align=c, scale=.1]{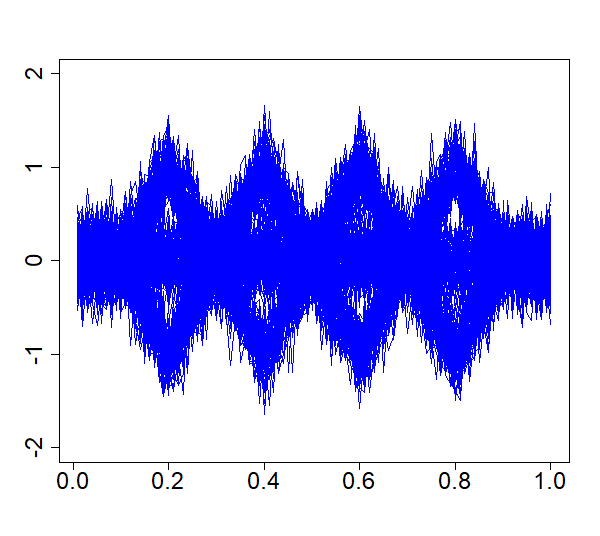} & \includegraphics[align=c, scale=.1]{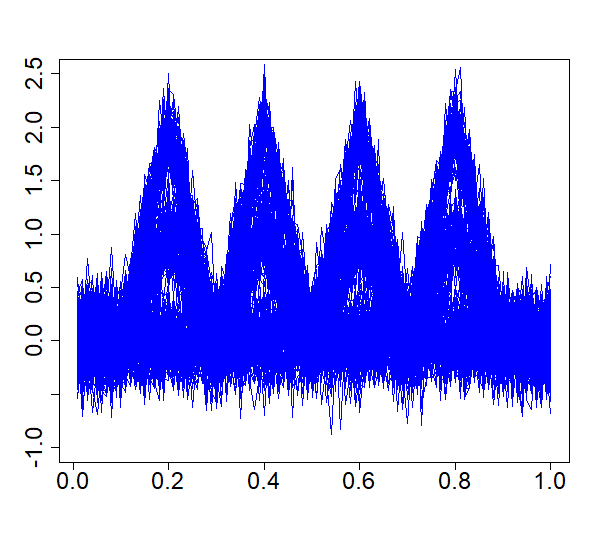}  
\end{tabular}\\  
$Y=1$ &  \begin{tabular}{c c}
     \includegraphics[align=c, scale=.1]{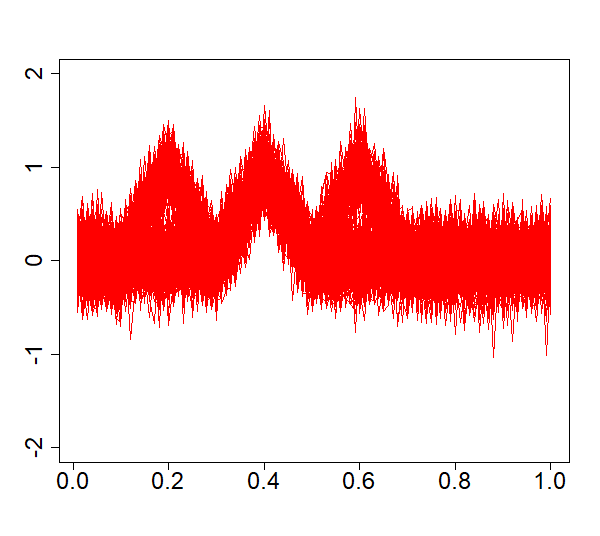} & \includegraphics[align=c, scale=.1]{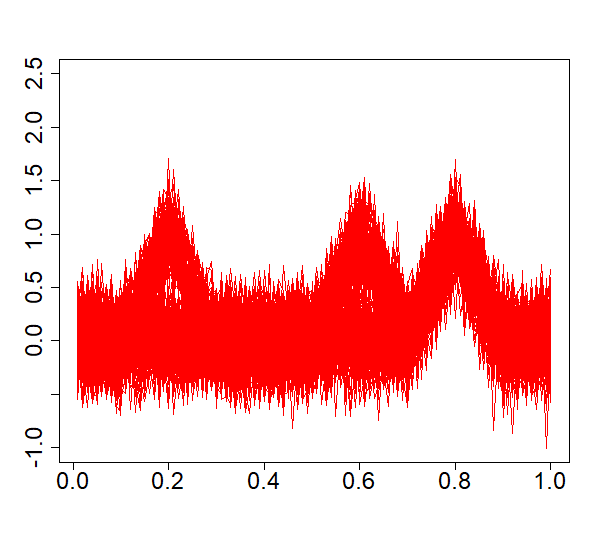}  
\end{tabular}  & \begin{tabular}{c c}
     \includegraphics[align=c, scale=.1]{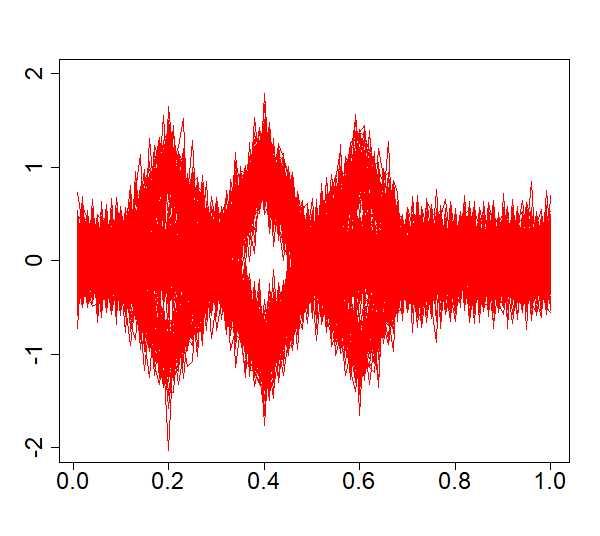} & \includegraphics[align=c, scale=.1]{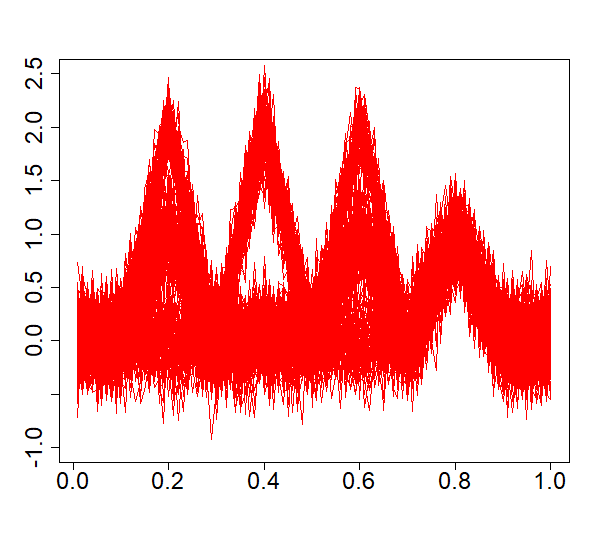}  
\end{tabular}   
    \end{tabular}
    \caption{ \small Example of curves $X$ in {Setting 2}, under both scenarios.
    Class 1 (red curves, $Y=1$) and class 0 (blue curves, $Y=0$). Scenario $2$ (right) shows more heterogeneity in class 1 compared to Scenario $1$ (left).}
    \label{data_plot}
\end{figure}

\begin{remark}  \text{ }
\begin{itemize}
\item In the two scenarios $\mathbb{P}(Y=1)=\mathbb{P}(Y=0)=0.5$. 
\item In Scenario $1$, the curves $X$ where $Y=1$ are composed exclusively of the realizations of events $V_1$ and $V_2$; in $X^{(1)}$, they have consecutive positive peaks at $t=0.2, 0.4$, or at $t=0.4, 0.6$.  
     \item Scenario $2$ is more complex. The four cases: $V_1$, $V_2$, $V_3$ and $V_4$ have the same probability to occur. In this case, the curves $X$ for $Y=1$ are more heterogeneous compared to Scenario $1$. The four cases, represented in Figure \ref{s_evn}, theoretically have the same proportion in the sample.  
\end{itemize}
\end{remark}
}
The functional form  of $X$ is reconstructed using 20 quadratic spline functions with equidistant knots. For a given scenario, we did 200  experiments. At each, 75 \% of the data are used for learning and 25 \% for validation. \\ The number of components for the MFPLS (in both models) is chosen  by 10-fold cross-validation. Moreover, MFPCA-LDA  is  performed for comparison. It consists, firstly in the estimation of principal components (using \cite{rMFPCA} package) and then applying linear discriminant analysis to them. As in the previous model, the number of components is chosen by 10-fold cross-validation. { \color{black} We also compute the approach proposed in \cite{MFPLS2022} (MFPLS$\_$D). To use it for classification, the transformation proposed in Section \ref{reg_class} is employed.  The number of components in MFPLS$\_$D is also chosen by $10$-fold cross-validation.}\\ 
 By defining the set of groups as  $\mathcal{G}_1=1, \mathcal{G}_2=2, \mathcal{G}_{3}=\{1,2\}$, the decision tree looks for the best splits among those obtained using separately univariate functions and the ones obtained using both dimensions. In order, to have an estimation of the optimal depth $m^*$, we randomly take  75\% of  learning data  to train an intermediate TMFPLS, and 25\% for pruning (by AUC metric).  This procedure is repeated  $10$ times and $\hat{m}^*$ is the frequent value among the $10$ repetitions. The final tree is then trained on the whole learning data, with the maximum tree depth fixed to $\hat{m}^*$, and the minimum criterion of impurity at 1\%. 

\subsubsection*{Results}
\begin{table}[ht]
\color{black}
\centering
\begin{tabular}{c c }
\rotatebox[origin=c]{270}{ \textbf{Scenario 1}}   & 
\begin{tabular}{crrrr}
 & \textbf{AUC} & \textbf{Sensibility} & \textbf{Specificity} & \textbf{Time(in s)} \\ 
TMFPLS & 0.98(0.02) & 0.97(0.03) & 0.99(0.02) & 39.48(9.49) \\ 
  MFPLS & 0.90(0.03) & 0.76(0.06) & 1.00(0.00) & 0.61(0.12) \\ 
  MFPCA-LDA & 0.90(0.03) & 0.76(0.06) & 1.00(0.00) & 0.63(0.13) \\ 
  MFPLS\_D & 0.90(0.03) & 0.76(0.06) & 1.00(0.00) & 0.06(0.02) \\ \end{tabular} \\ \hline 
    \rotatebox[origin=c]{270}{ \textbf{Scenario 2}}  & \begin{tabular}{crrrr}
  
 & {\color{white}\textbf{AUC}}& \textbf{\color{white}Sensibility} & \textbf{\color{white}Specificity} & \textbf{\color{white}Time(in s)}
 \\ 
  TMFPLS & 0.97(0.02) & 0.96(0.04) & 0.98(0.03) & 60.74(125.90) \\ 
  MFPLS & 0.50(0.06) & 0.49(0.07) & 0.52(0.10) & 0.54(0.13) \\ 
  MFPCA-LDA & 0.56(0.11) & 0.53(0.20) & 0.56(0.22) & 0.57(0.13) \\ 
  MFPLS\_D & 0.50(0.06) & 0.50(0.06) & 0.50(0.07) & 0.05(0.01) \\ 
  \end{tabular}
\end{tabular}
\caption{\color{black} Results of Setting $2$.  For the two scenarios, the presented metrics (AUC, Sensibility, etc.) are the means and the standard deviations (in parentheses) of those obtained in the experiences.}
\label{res_cv}
\end{table}
{\color{black}
In Scenario $1$,  Table \ref{res_cv}  shows that AUC differences are about $8$\% between MFPLS and TMFPLS. Furthermore, MFPCA-LDA and MFPLS$\_$D  are competitive with MFPLS. Table \ref{res_cv} also exhibits that the training procedure of TMFPLS is time-consuming compared to the other methods, about $66$ times more than the time for training MFPLS. This can be explained by the fact that at each split in the tree, we used a cross-validation procedure to choose the number of components. {\color{black} Also, the results clearly show that MFPLS$\_$D is the fastest method. }    
} 

Scenario $2$ shows more differences between the methods: MFPLS, {\color{black}MFPLS$\_$D} and MFPCA-LDA are non-effective compared to TMFPLS. Hence, TMFPLS outperforms these methods in a complex task classification such as Scenario $2$. {\color{black} It is worth noting that in this case, the (mean) time for the estimation of TMFPLS has significantly increased compared to Scenario $1$.   \par 
This is because the estimated trees in Scenario $1$ have fewer ramifications than the estimated ones in Scenario $2$. As an illustration, we can refer to Figure \ref{scen_1} and Figure \ref{sc_2_tree}, which represent examples of trees estimated respectively from Scenario $1$ and Scenario $2$. These trees are randomly selected among the $200$ estimated for each scenario. Moreover, Figure \ref{split_b} and Figure \ref{split_b2} present the associated discriminant functional coefficients used for splitting rules in the trees. We provide insight on how to interpret them.  \par   
{\color{black} In Scenario $1$, the tree uses in the first splitting rule (depth=$0$) a coefficient function that only depends on the first dimension 
. The fact that this function has a negative peak at $t=0.4$ can be interpreted as if a curve $X$ has a low value of $X^{(1)}$ at this region, it will belong to the left node of depth$=1$. Note that this node has exclusively curves of class $Y=0$. This makes sense since in Scenario $1$, $X$ curves where $Y=1$ are defined as having peaks in $X^{(1)}$ at points $(0.2, 0.4)$ or $(0.4, 0.6)$. In other words, a high value at $t=0.4$ characterizes the class $Y=1$. \\ 
At the depth$=1$, the coefficient function depends only on the second dimension. The resulting separation doesn't lead to total purity in the node, however, the coefficient states that curves $X$ in the right node with high values of $X^{(2)}$ at $(0.2, 0.4)$ are more frequent in class $Y=0$. This last rule is not senseless, since Figure \ref{s_evn} shows that curves $X$ for $a=V_1$ don't have peaks in $X^{(2)}$ at $t=0.2$ and $t=0.4$. \\ 
Using only these two coefficients, we can say that for curve $X$ if at $t=0.4$ $X^{(1)}$ doesn't have a peak, $X$ is of class $Y=0$. However, if $X^{(1)}$ has a peak at $t=0.4$, and $ X^{(2)}$ has \textit{small} values at regions $(0.2, 0.4)$, $X$ is probably of class $Y=1$. The other coefficient functions help to refine the predictions by giving more and more precise rules at each step.    }

\begin{figure}[H]
    \centering
    \begin{tabular}{c c}
       \includegraphics[width=.8\linewidth, align=c]{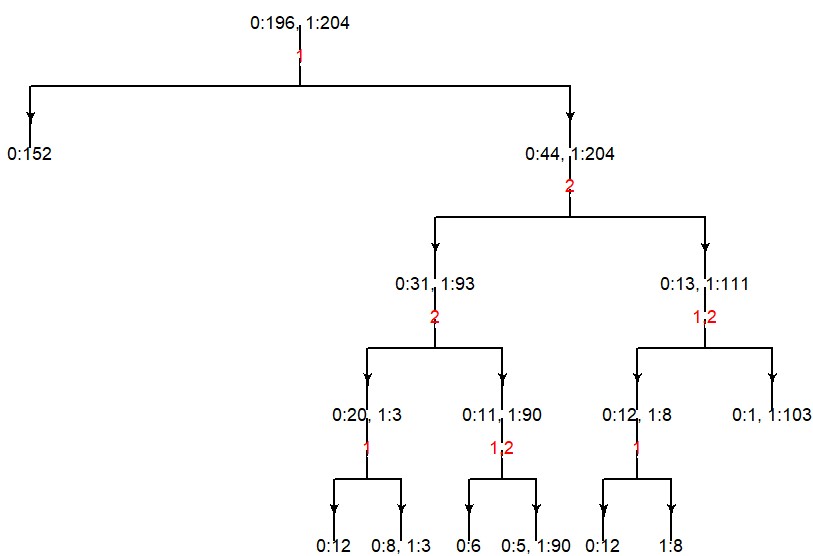}   &  
    \end{tabular}
    \caption{ \small Example of post-pruned tree  in Scenario $1$(Setting 2). Dimension(s) used for the splitting is in red.
    }  
    \label{scen_1}
\end{figure}

\begin{figure}[H]
\centering
\begin{tabular}{ c|   c  c c c c  c c  }
     \textbf{depth} &  $\to$ & $\to$ \\
     \hline 
     \textbf{0}& \includegraphics[align=c, scale=.07]{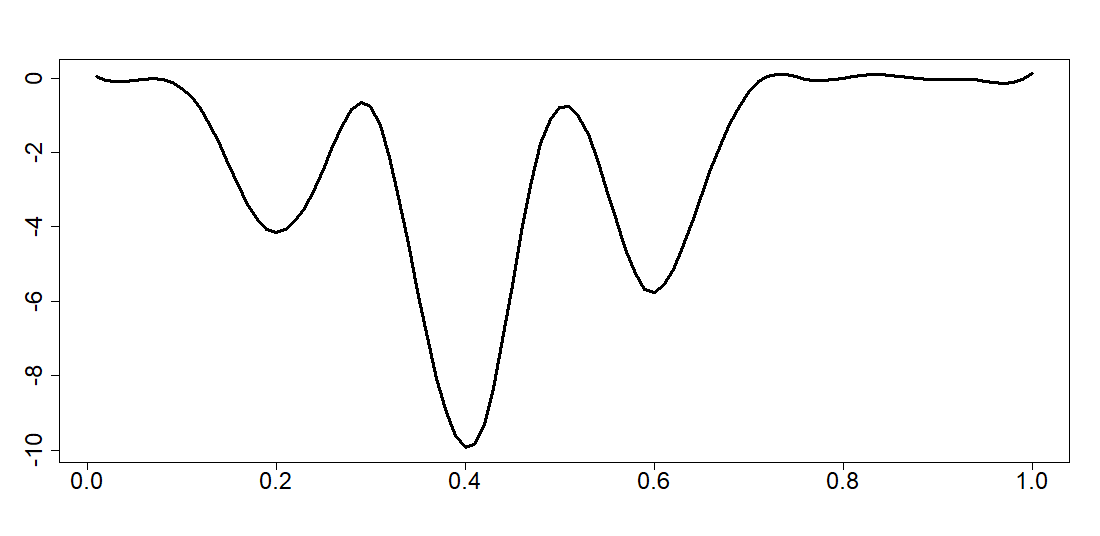} &    \\
     &  & \\ 
     \textbf{1} &\includegraphics[align=c,scale=.07]{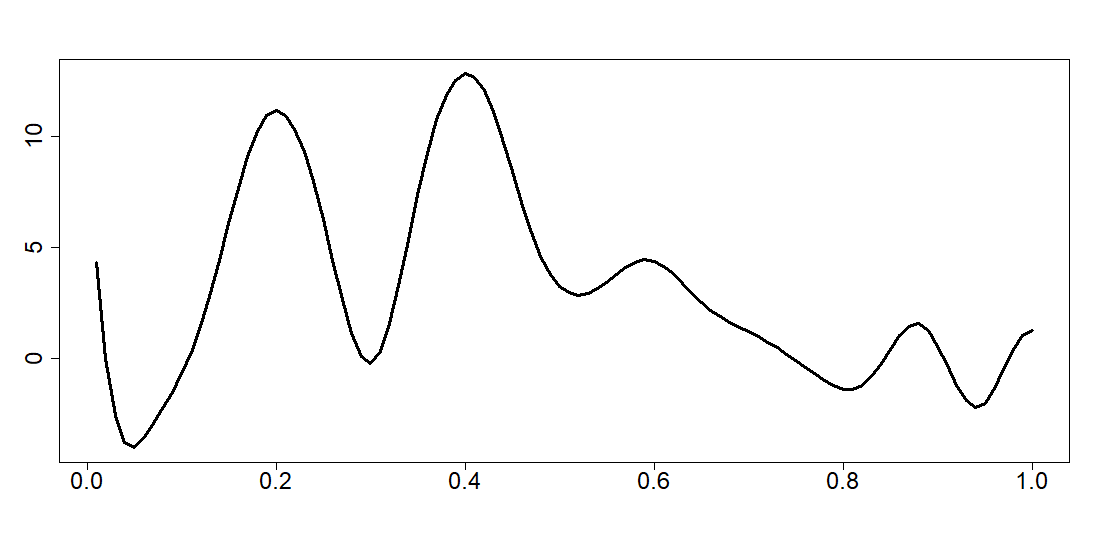}   \\ 
     \textbf{2} & \includegraphics[align=c,scale=.07]{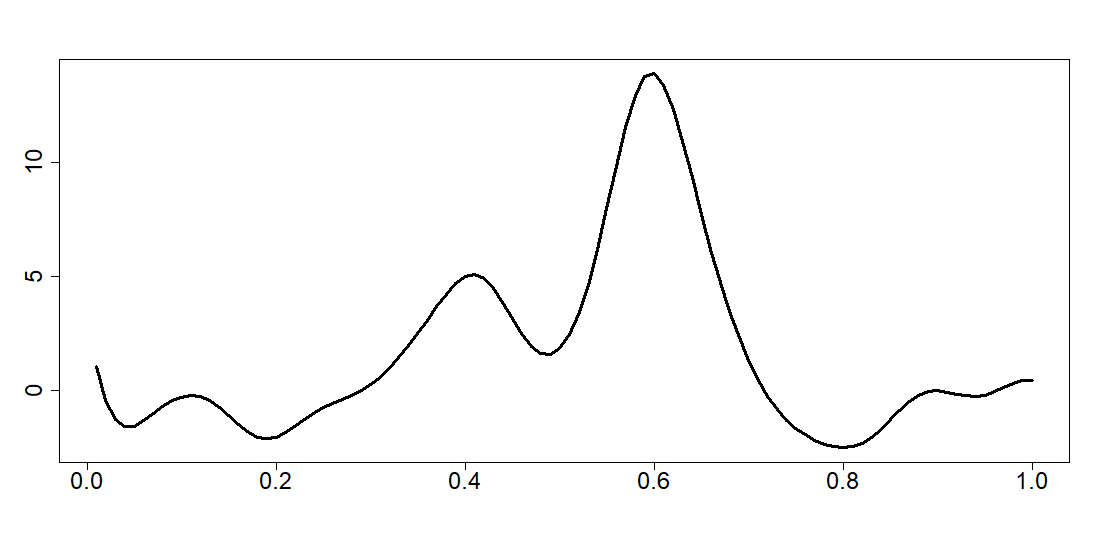} & \includegraphics[align=c,scale=.07]{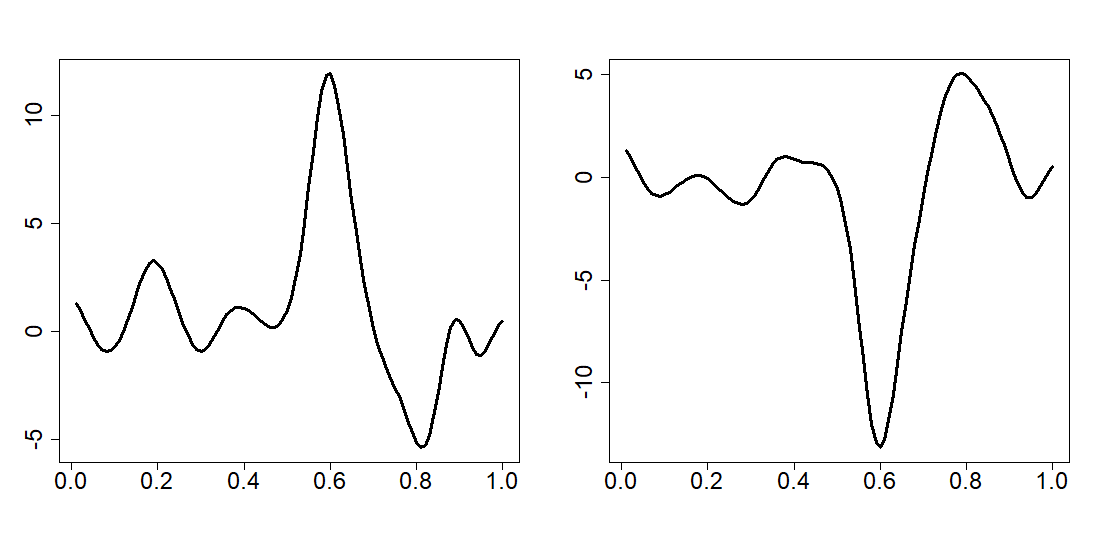}  \\ 
    \textbf{3} & \includegraphics[align=c,scale=.07]{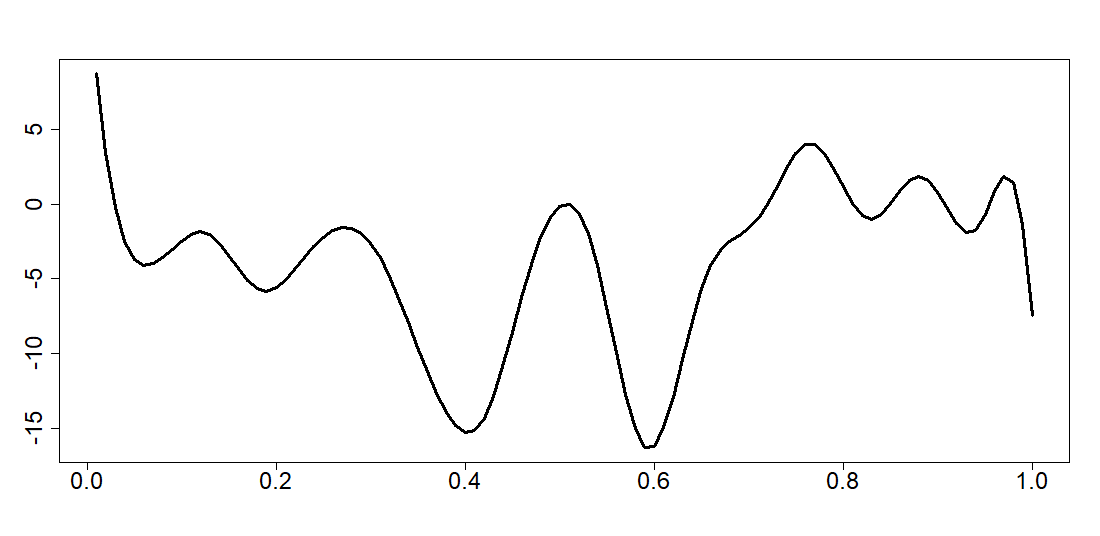} & \includegraphics[align=c,scale=.07]{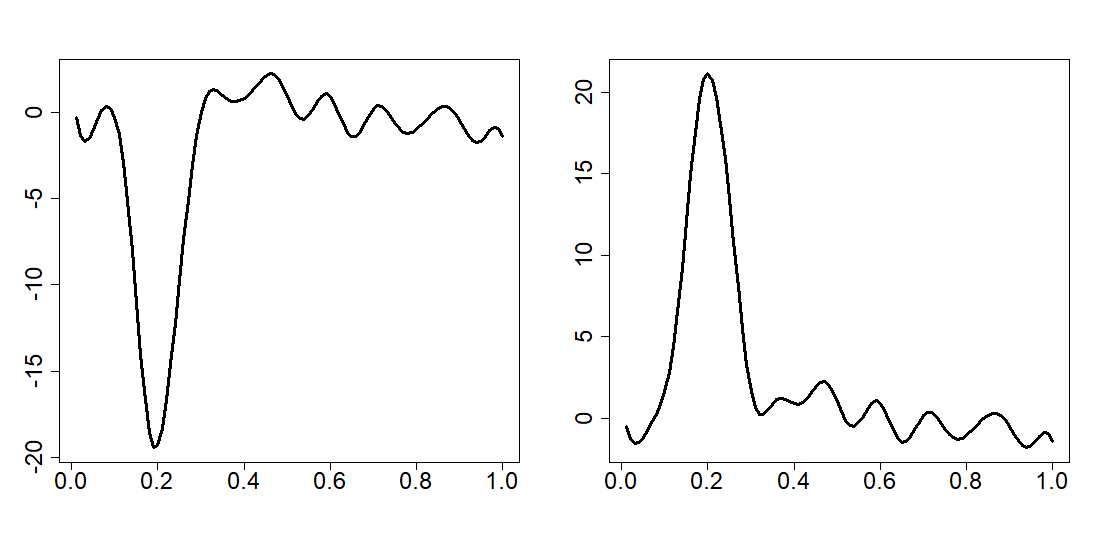} & \includegraphics[align=c,scale=.07]{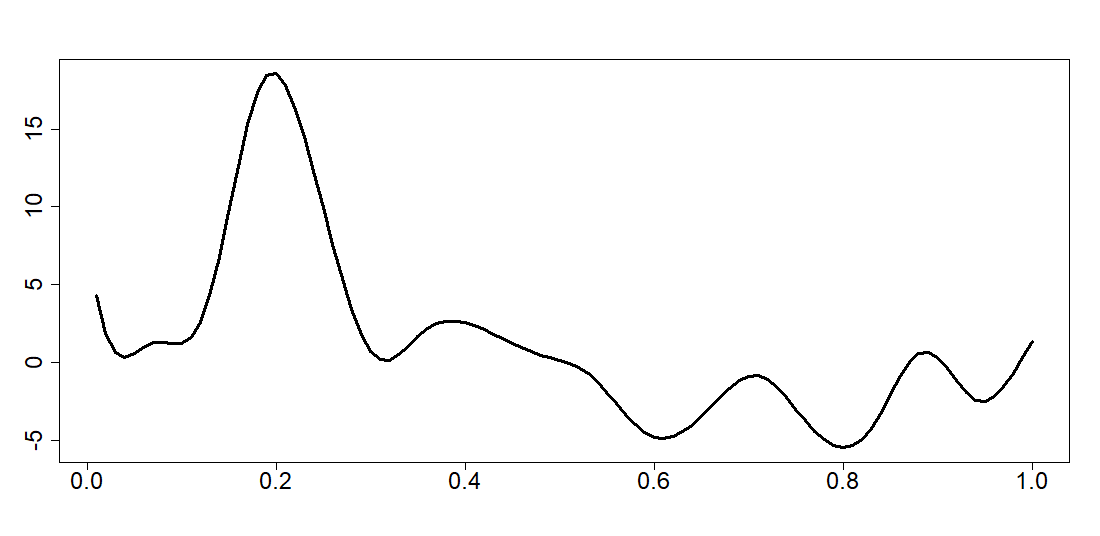}    
\end{tabular}
\caption{Scenario 1: TMFPLS  estimated coefficient functions (from left to right) at each split}
\label{split_b}
\end{figure}
\par In the scenario $2$, the goal is to estimate a more complex rule of classification, since the amplitudes ($a_1, a_2, a_3, a_4$) can be either positive or negative for curves $X$ of class $Y=1$. As already mentioned, compared to Scenario $1$, the example estimated tree (Figure \ref{sc_2_tree}) has more ramifications. As in the classical tree, it is challenging to give a clear interpretation of TMFPLS in a such case,  especially, when using only Figure \ref{sc_2_tree} and Figure \ref{split_b2}. There are numerous coefficient functions that we have to take into account. One need, which we should address in the future, is to provide more insightful visualizations in this case. However, this tree shows how TMFPLS can be flexible and can therefore be used in a complex classification framework.
}
\begin{sidewaysfigure}
\centering     
\includegraphics[width=\linewidth,align=b]{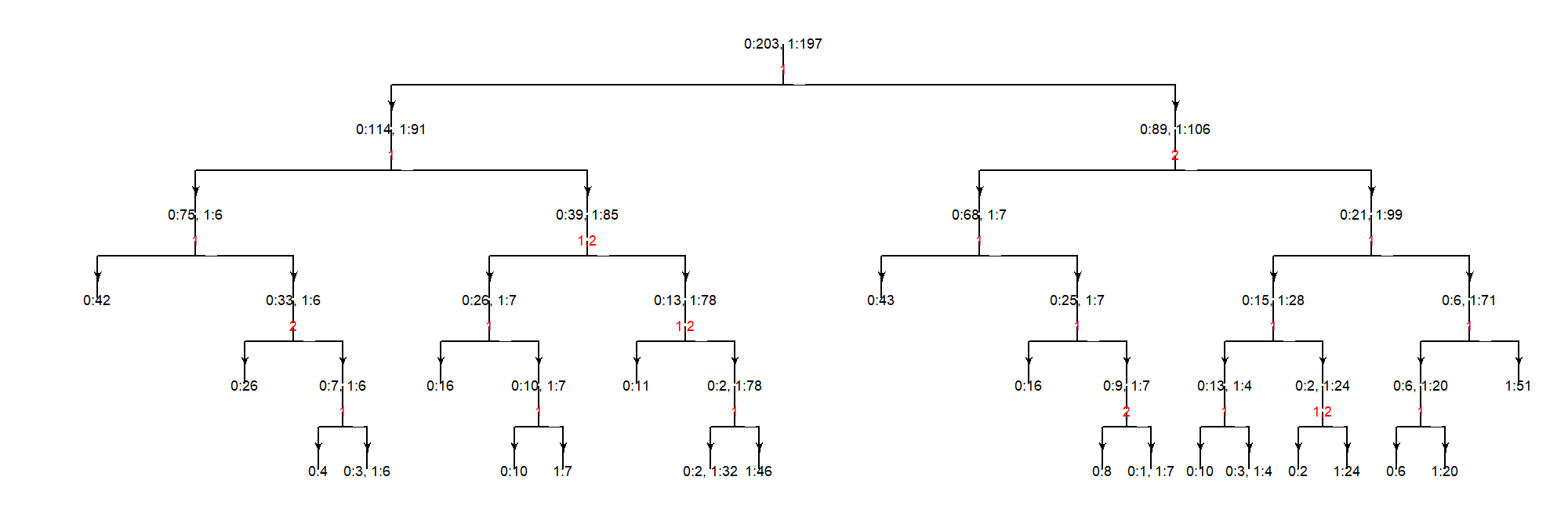} 
          \caption{ \small Example of post pruned tree  in Scenario $2$( {Setting 2}).
          }
          \label{sc_2_tree}
\end{sidewaysfigure}

\label{output}

\begin{sidewaysfigure}
\centering
\begin{tabular}{c| c c c c c c  c c c}
\textbf{depth}&  $\to$ & & $\to$ &  & $\to$ &   \\ 
 \hline \\ 
   \textbf{0}  & \includegraphics[align=c,scale=.06]{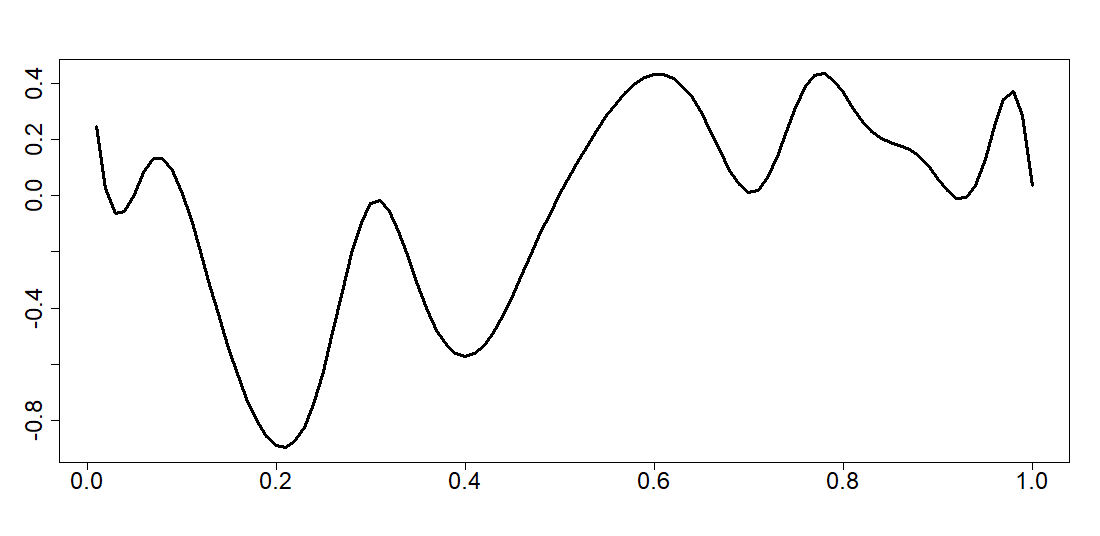}  \\
   \textbf{1}   & \includegraphics[align=c,scale=.06]{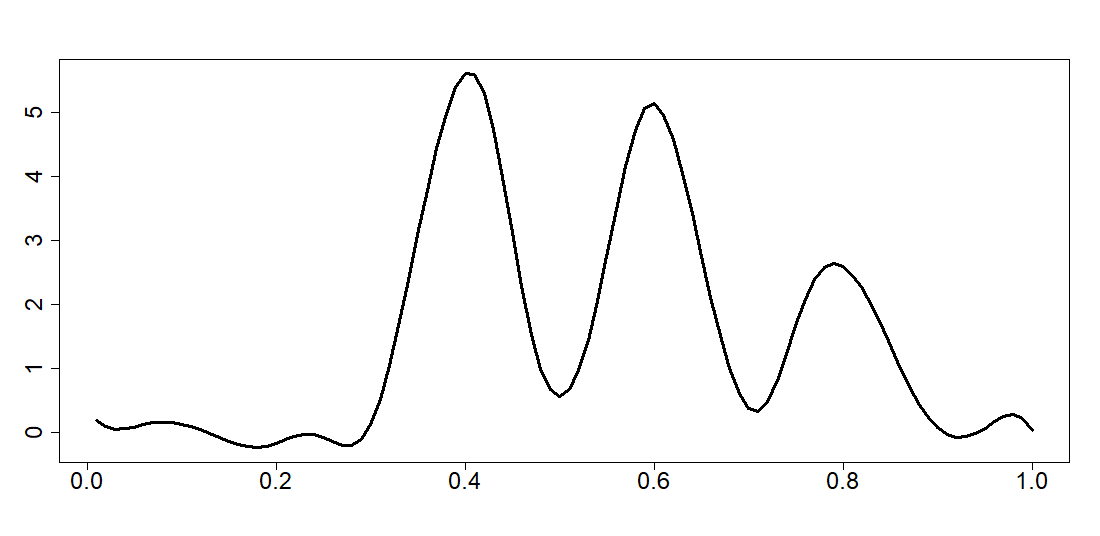} & \includegraphics[align=c,scale=.06]{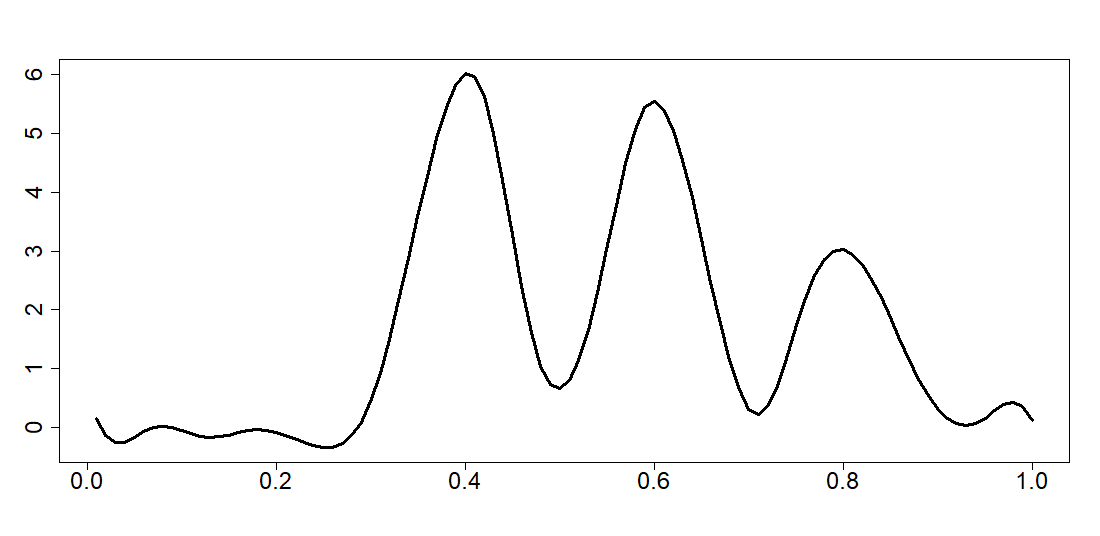} \\
      \textbf{2}& \includegraphics[align=c,scale=.06]{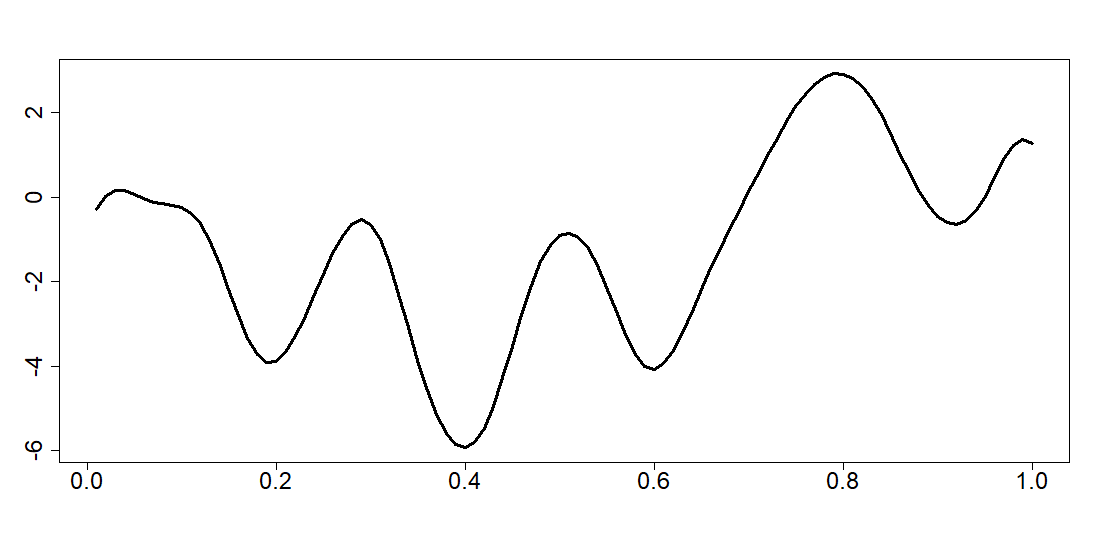} & \includegraphics[align=c,scale=.06]{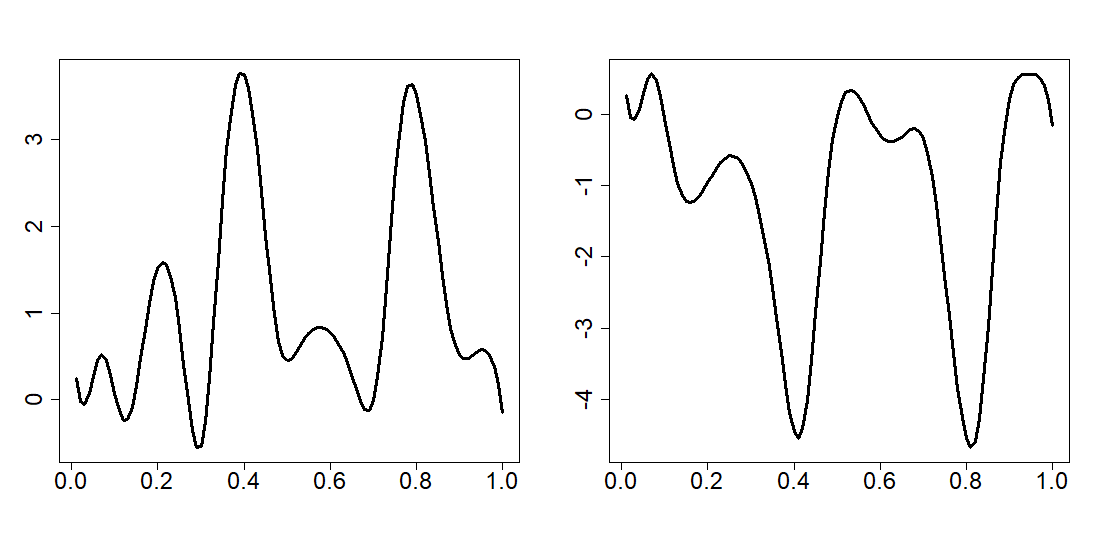} & \includegraphics[align=c,scale=.06]{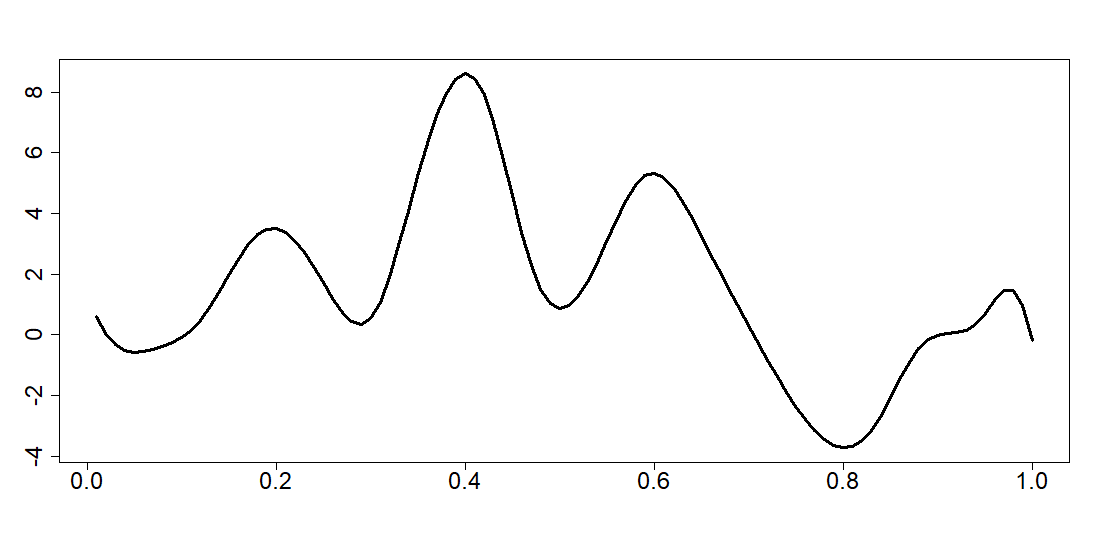}& \includegraphics[align=c,scale=.06]{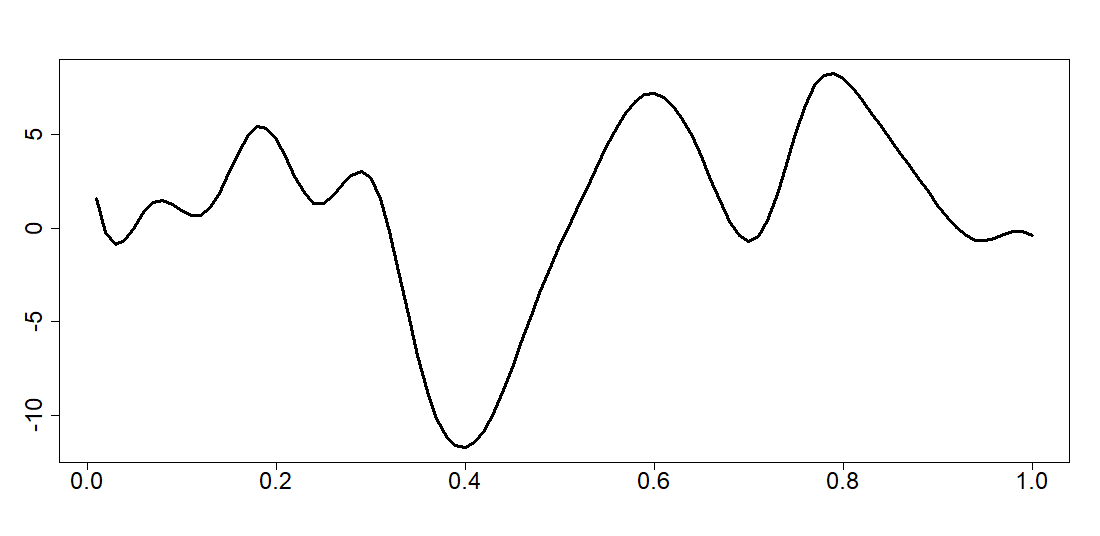}  \\ 
         \textbf{3}&  \includegraphics[align=c,scale=.06]{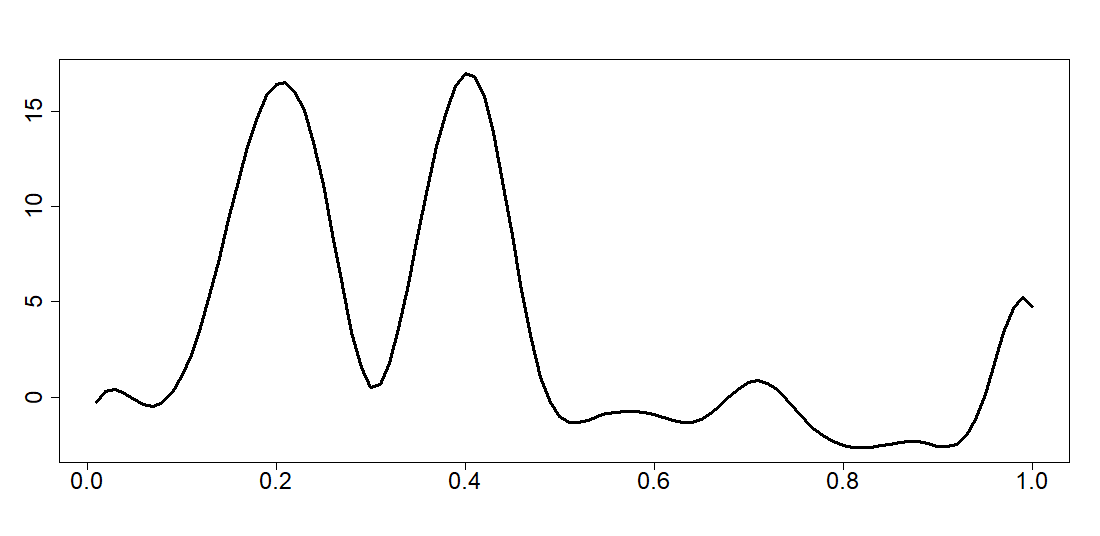} & \includegraphics[align=c,scale=.06]{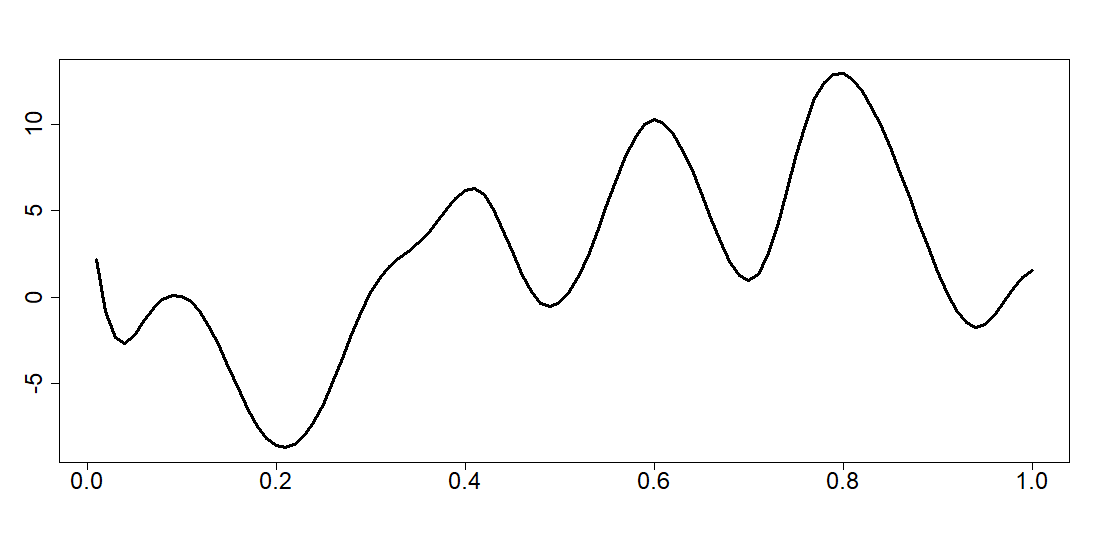} & \includegraphics[align=c,scale=.06]{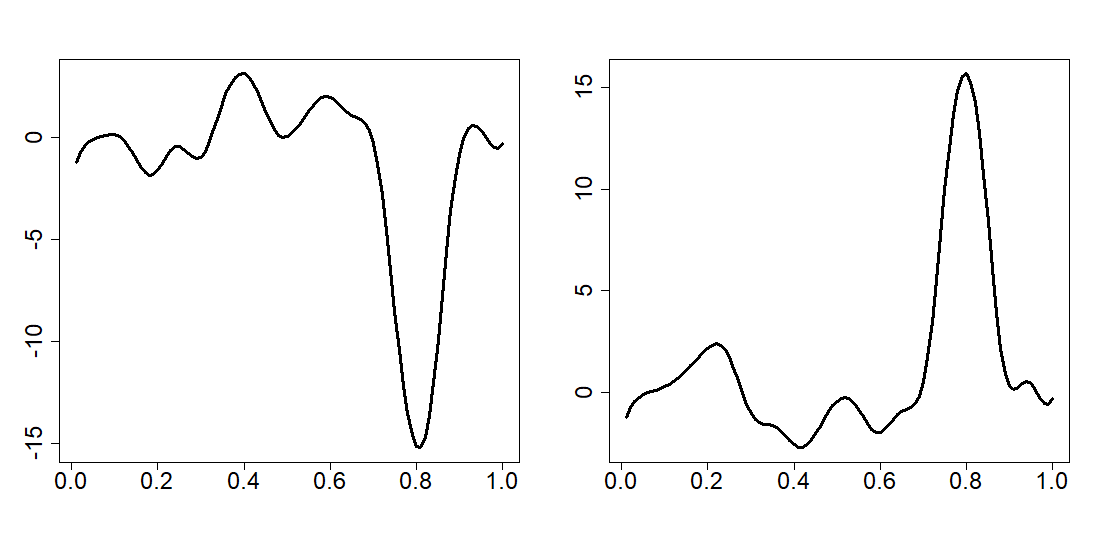} & 
     \includegraphics[align=c,scale=.06]{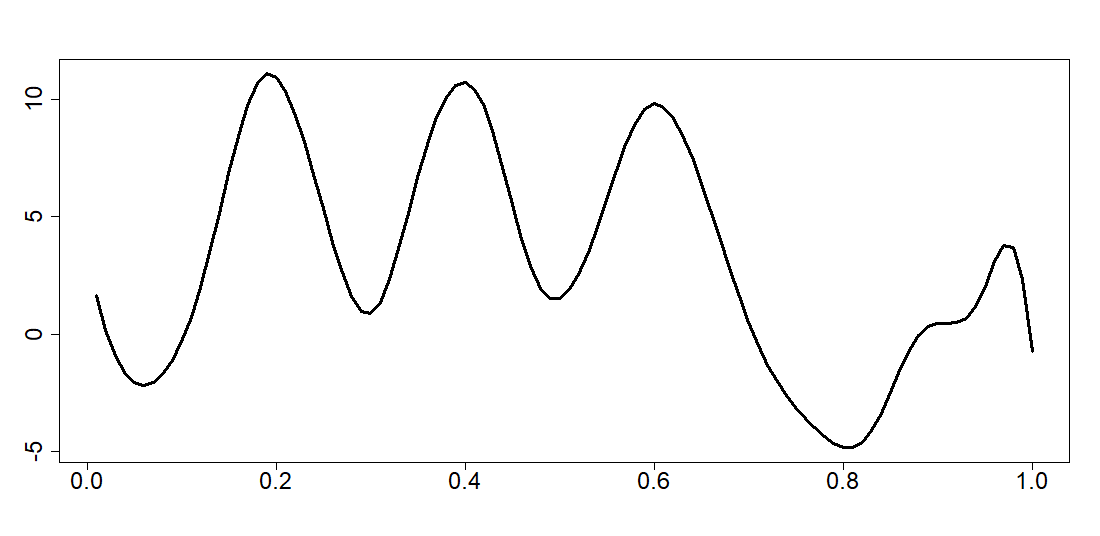} & 
     \includegraphics[align=c,scale=.06]{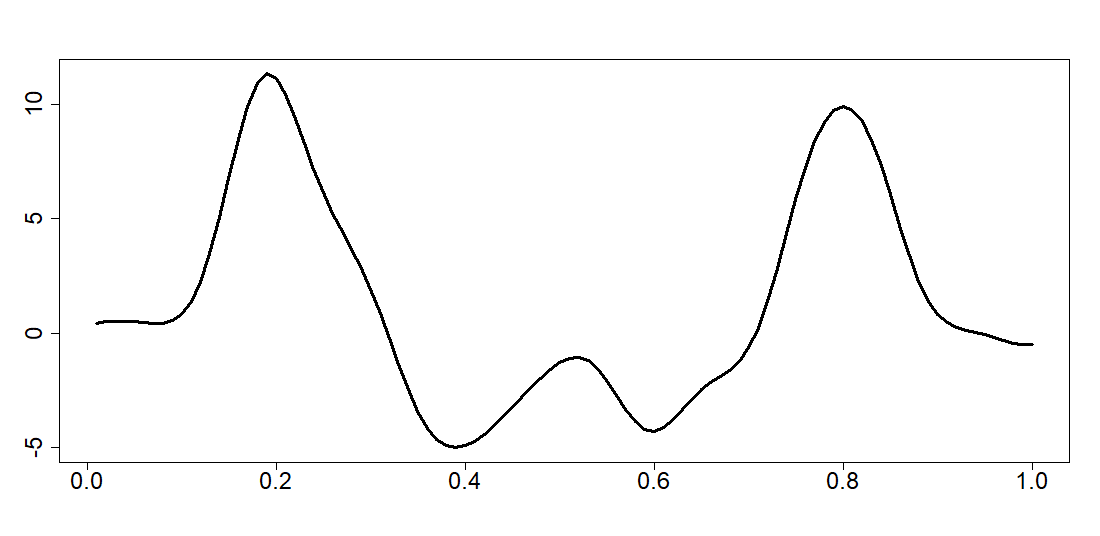} &      \includegraphics[align=c,scale=.06]{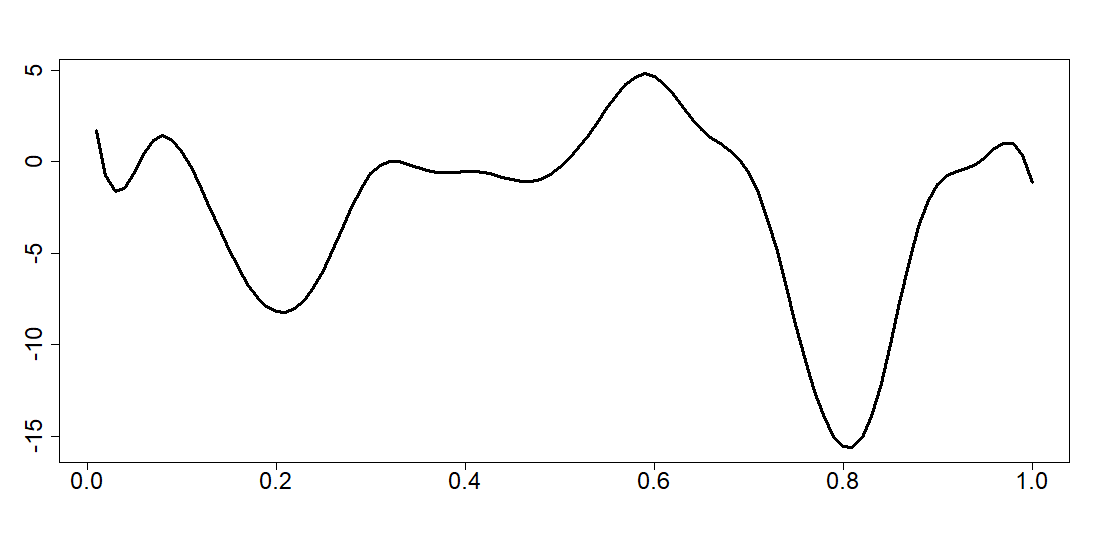} \\ 
    \textbf{4} &  \includegraphics[align=c,scale=.06]{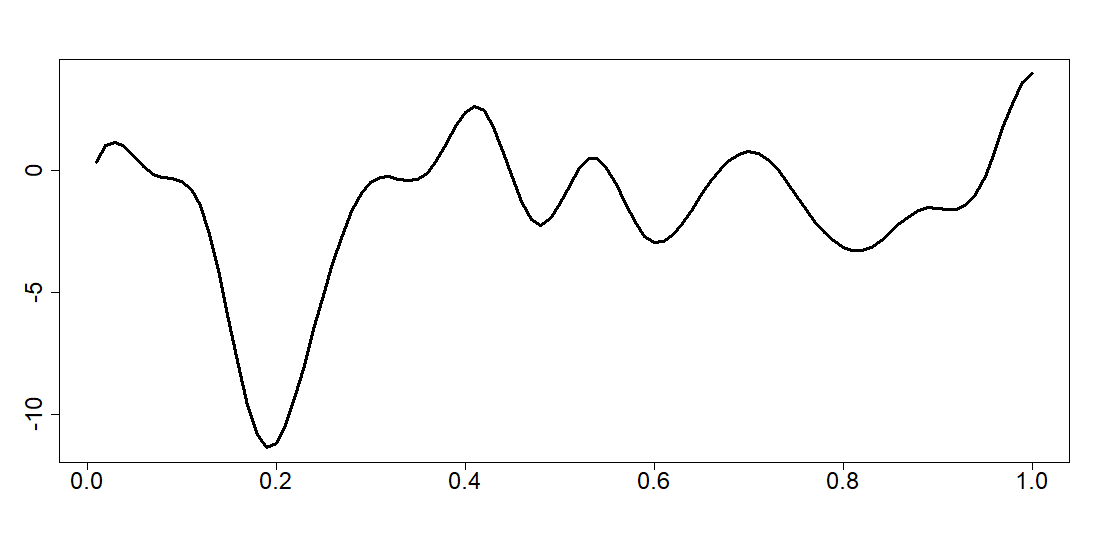} &  \includegraphics[align=c,scale=.06]{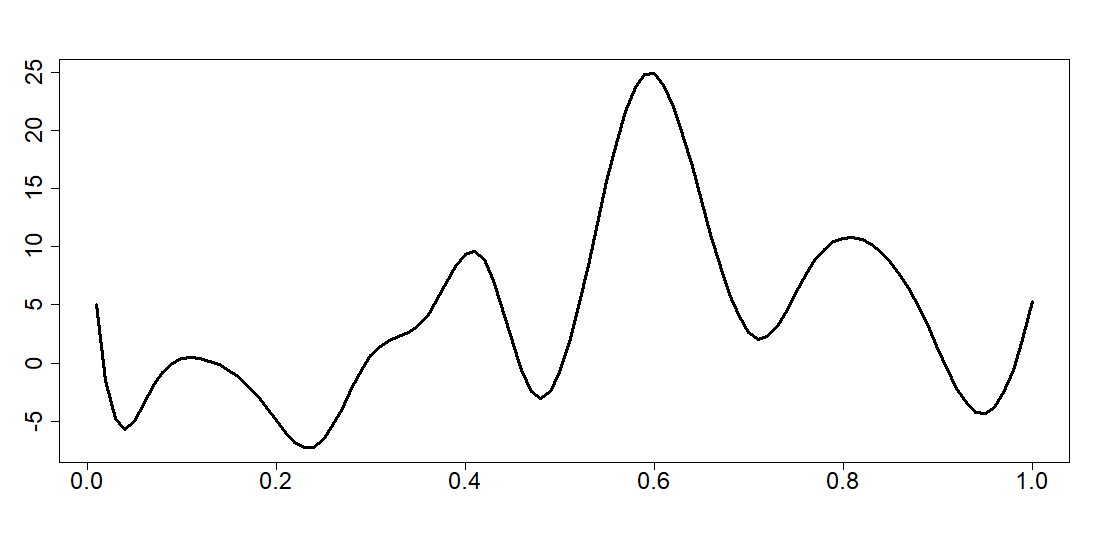}& \includegraphics[align=c,scale=.06]{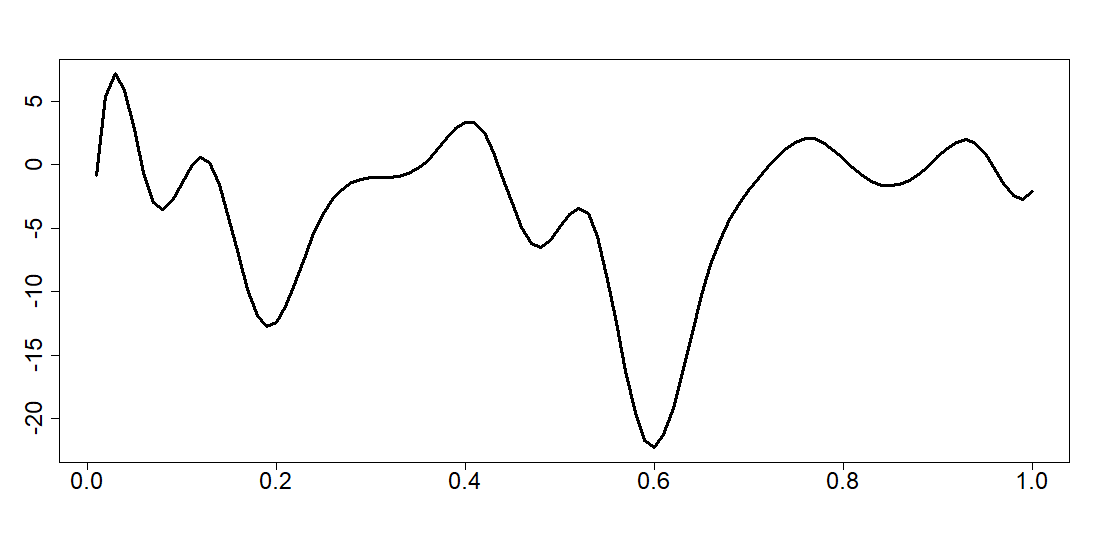} &\includegraphics[align=c,scale=.06]{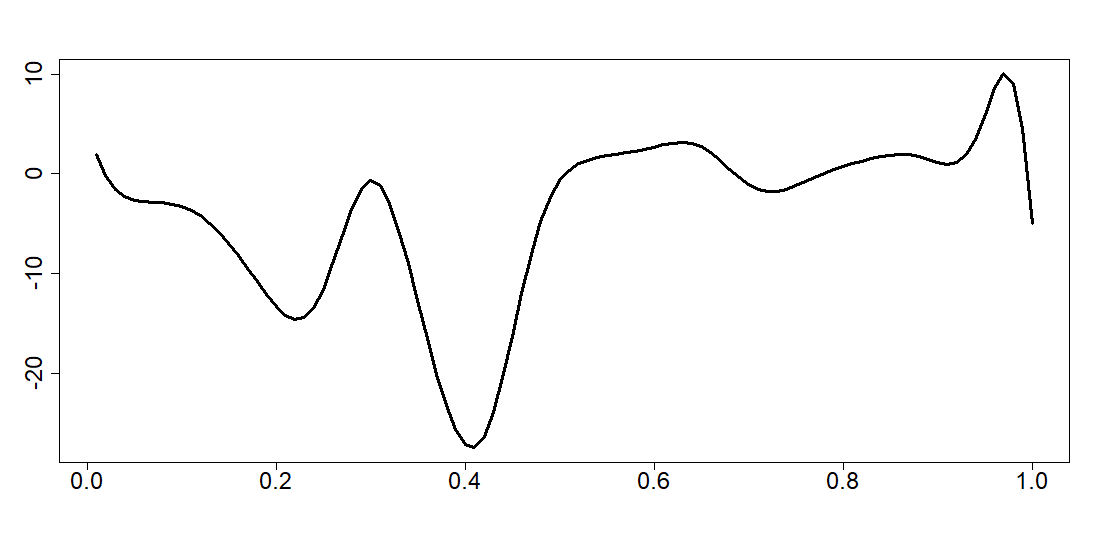} & \includegraphics[align=c,scale=.06]{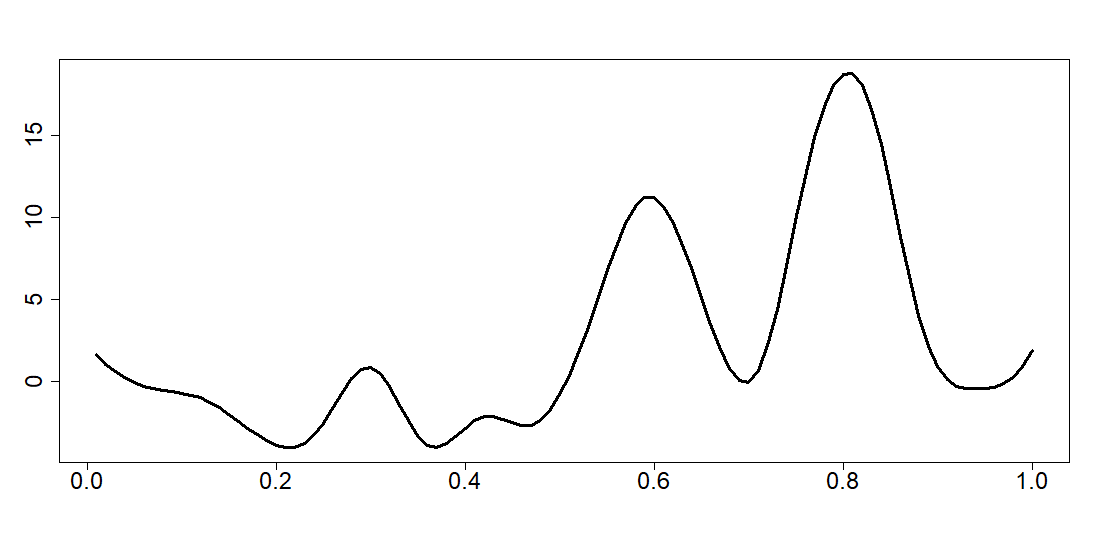} & \includegraphics[align=c,scale=.06]{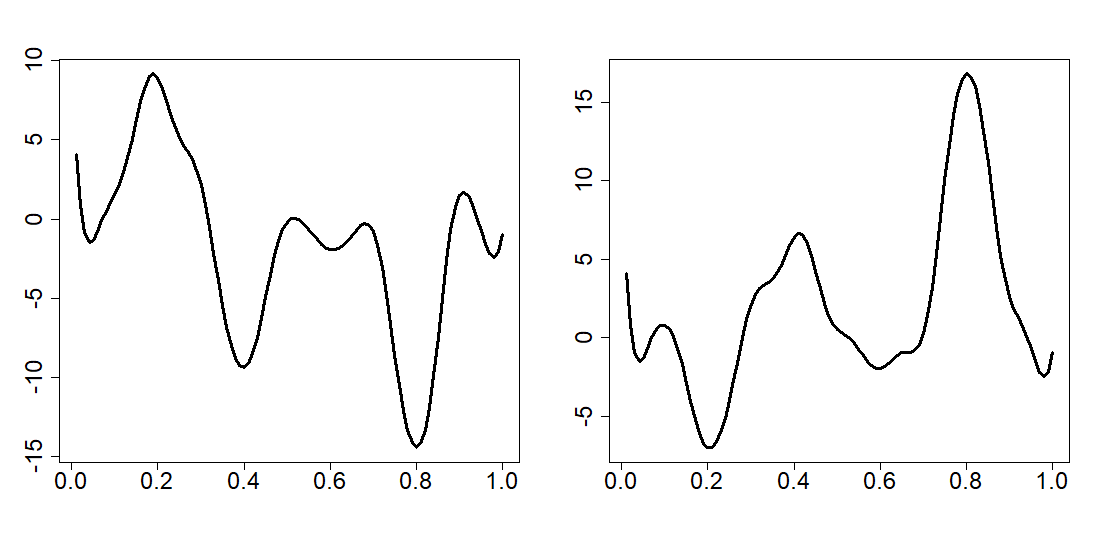}& \includegraphics[align=c,scale=.06]{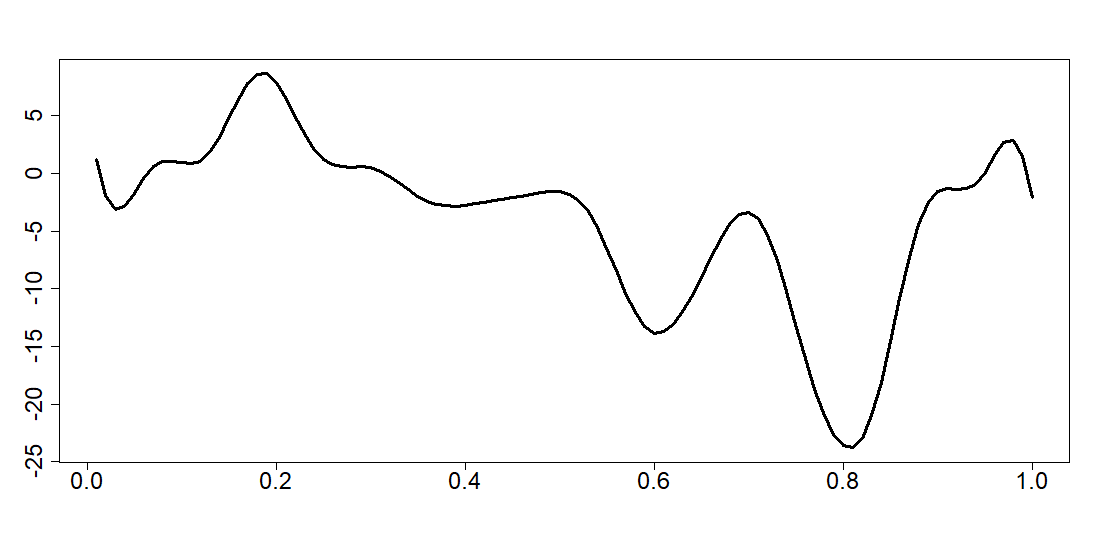}     \\
\end{tabular}
\caption{Scenario 2: Coefficient functions in  the estimated tree 
}
\label{split_b2}
\end{sidewaysfigure}
}

\par 

\subsection{Different domains case}
\subsubsection{ {Setting 3:} Image and Time series classification}
Our approach allows  the use of  images and time series simultaneously. In this part, we highlight the use of various domains instead of focusing only on one dimension domain.
\paragraph{Framework} 
Consider the domains $\mathcal{I}_1=[0, 50]$,  $\mathcal{I}_2=[0,1] \times [0,1]$, and  $X=(X^{(1)}, X^{(2)})^{\T1} $:  
\begin{align*}
    X^{(1)}(t)&=Z_1h(t) + \epsilon^{(1)}(t)\text{, }  t \in \mathcal{I}_1  \\ 
    X^{(2)}(t) & = Z_2q(t) + \epsilon^{(2)}(t) \text{, }t \in \mathcal{I}_2.
\end{align*}
The  noise term $\epsilon=(\epsilon^{(1)}, \epsilon^{(2)})^{\top}$ is composed of two independent dimensions: the first one   $\epsilon^{(1)}$ is a  white noise of variances $\sigma^2$, while the second one  $\epsilon^{(2)}$ is a gaussian random field. $\epsilon^{(2)}$ is associated with a Matern covariance model, with  sill, range, and nugget parameters equal to  $0.25$, $0.75$, and $\sigma$, respectively (see \cite{geor} for more details). The variables $Z_1$ and $Z_2$ are Bernoulli variables with values in $\{0,1\}$. The (deterministic) functions $h$ and $q$ are given by:  
\begin{align*} \; h(t) &= 
        3.14 \left(1-\frac{\vert t-10\vert}{4}\right)_+   &   q(s) &= -2\log \left(\sqrt{(s^{(1)}-0.5)^2 + (s^{(2)}-0.5)^2} \right) 
\end{align*}
where $(.)_+$ denotes the positive part, $t\in \mathcal{I}_1$,  $s=(s^{(1)}, s^{(2)}) \in \mathcal{I}_2 $.\\
The response variable $Y$ is constructed as follows : 
\begin{equation*}
Y = \left\{
    \begin{array}{ll}
        1 & \mbox{if } 
    Z_1Z_2 =1  \\ 
        0 & \mbox{ otherwise.}
    \end{array}
\right.
\end{equation*}

In other words, $Y=1$ if and only if both variables $Z_1$, $Z_2$ are simultaneously 1 (see Figure \ref{y_1}).
 \begin{figure}[h]
 \centering
    \begin{tabular}{ c c c  c  c c}
    $X^{(1)}(t)$ & \includegraphics[align=c, scale=.07]{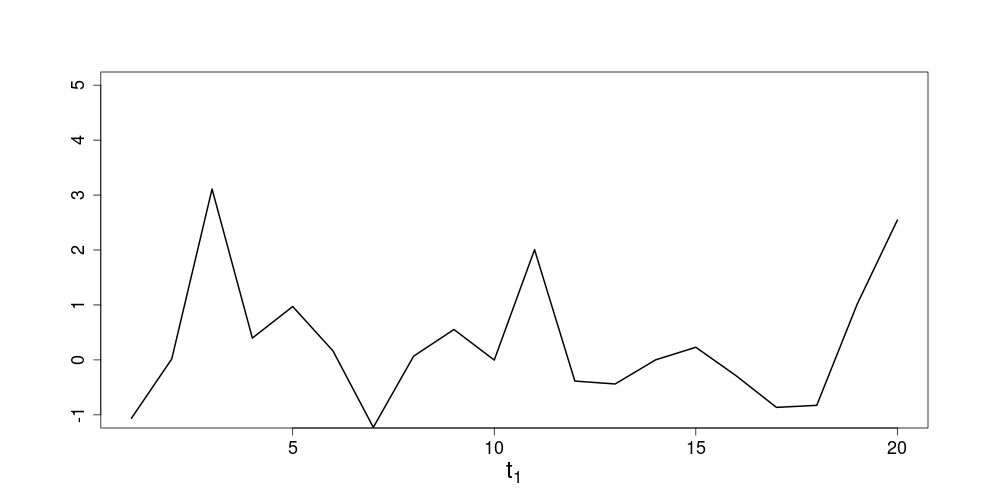}& \textbf{+} &   
        \includegraphics[align=c, scale=.07]{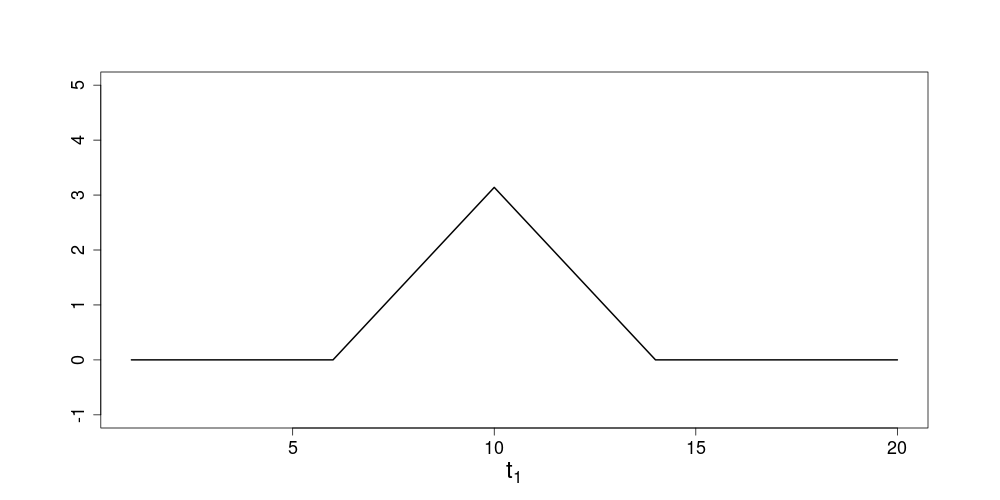}
        & \textbf{=} &  \includegraphics[align=c, , scale=.07]{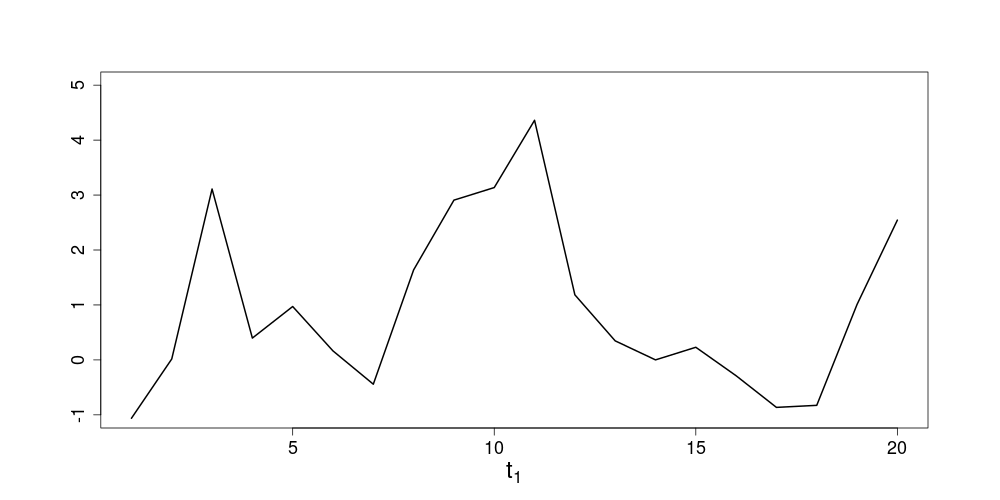} 
        \\ 
         $X^{(2)}(t) $ & \includegraphics[align=c, , scale=.07]{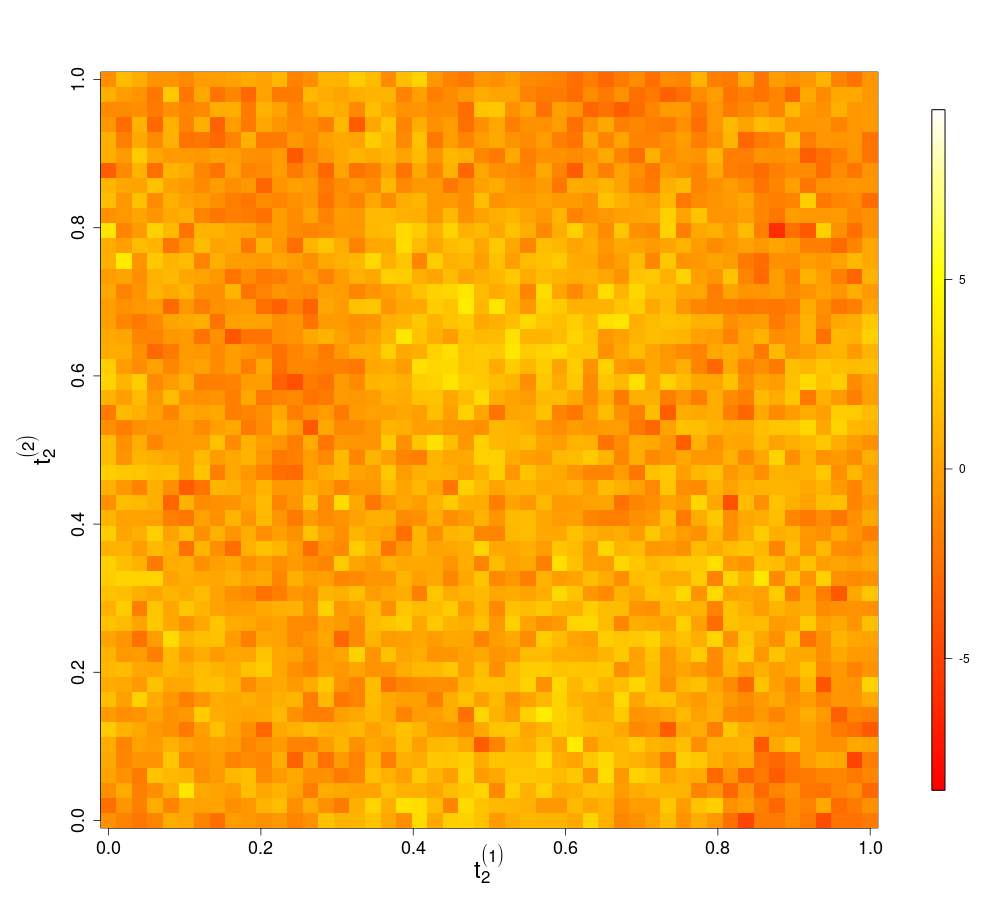}& \textbf{+} &   
        \includegraphics[align=c, , scale=.07]{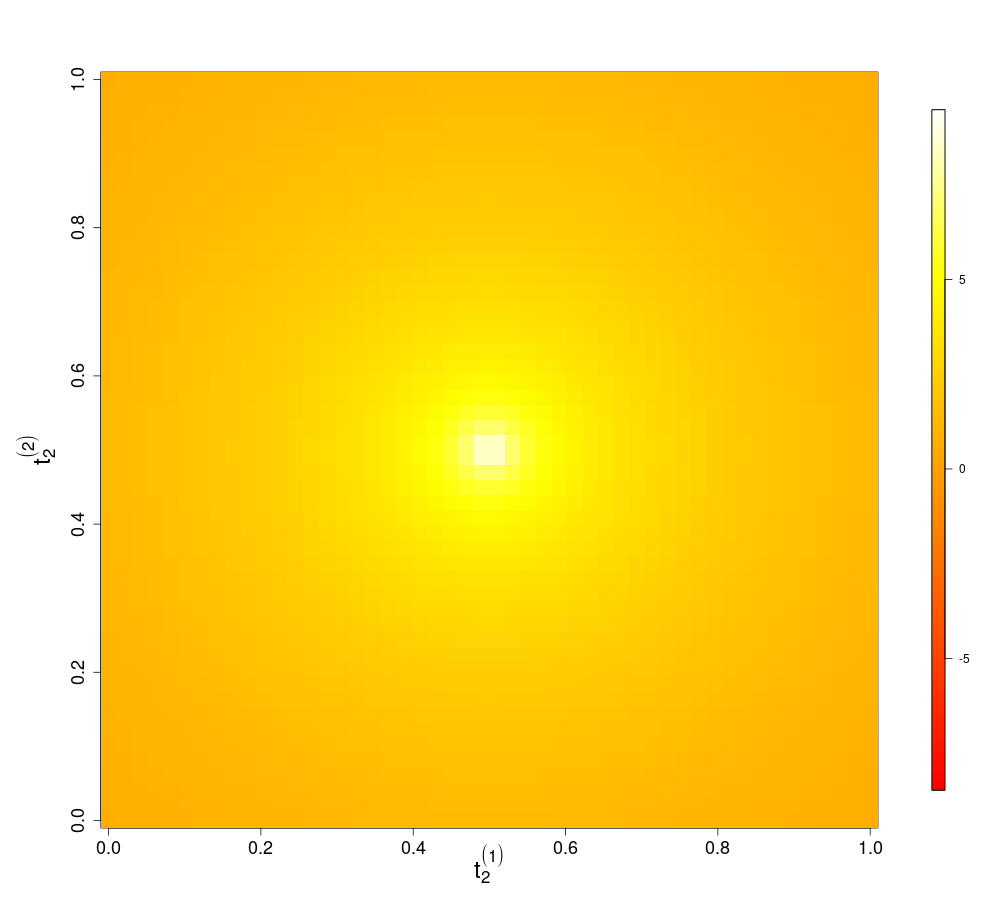}
        & \textbf{=} &  \includegraphics[align=c, , scale=.07]{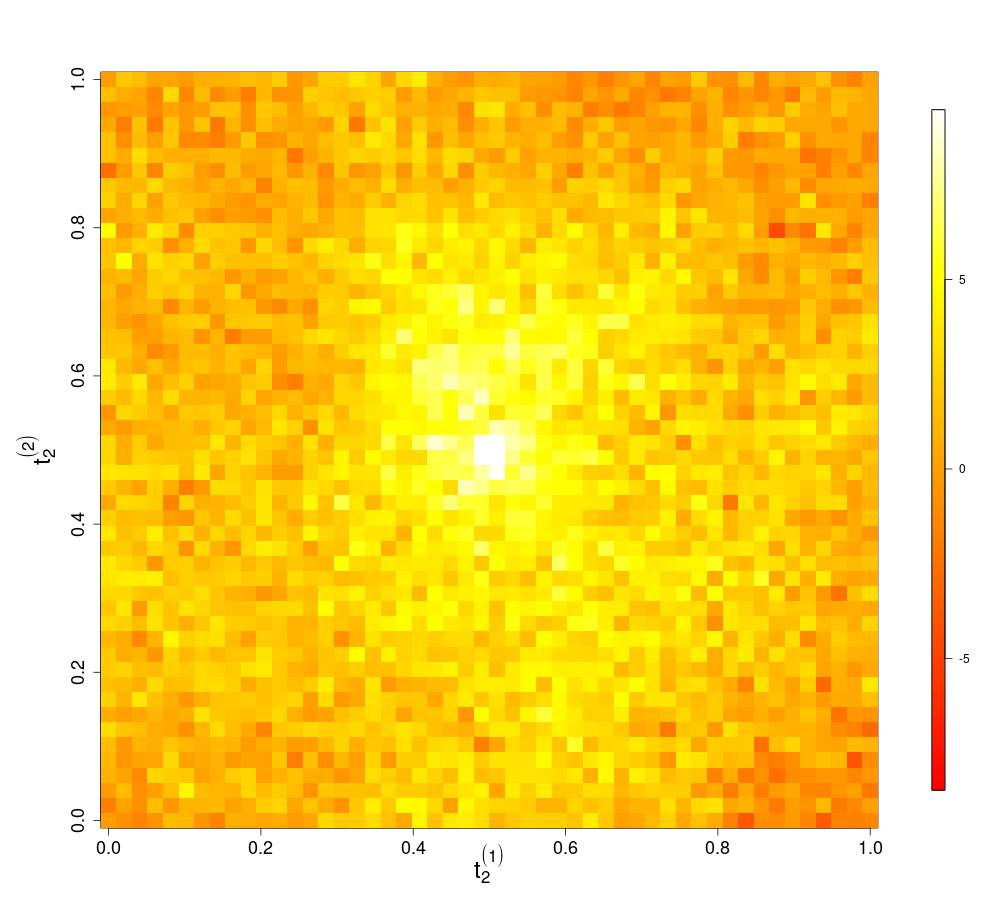} 
    \end{tabular}
    \caption{  Construction of class 1 ($Y=1$), example of curve $X(t)$ in {Setting 3}, under SNR=$0.5$  }
    If $Y$=0, $X$ is random noise $\epsilon$ (left figures).\\ 
    \label{y_1}
 \end{figure}
 \par  $50$ equidistant discrete points and $50\times 50$  pixels are observed respectively for the first and the second dimension.  
To get the functional form of $X$, the first and second components are projected respectively into the space spanned by $20$ quadratic spline functions, and the 4 two-dimensional splines \citep{rMFPCA}. \\ 
The variances of the functions $q$ and $h$ along their domain are approximately 
1. The signal-to-noise ratio (SNR) is then (approximately) the same on both dimensions and depends only on $\sigma$  
 \begin{equation*}
     \text{SNR}=\frac{1}{\sigma^2}. 
 \end{equation*}
By controlling the parameter $\sigma$, we consider several values of SNR: $0.5$, $0.7$, $1.2$, $2.1$ and $4.9$ .\par
We set $\mathbb{P}(Z_1=1)=\mathbb{P}(Z_2=1)=3/4$, then  $\mathbb{P}(Y=1) = 9/16\simeq 0.56$. A set of 500 curves are simulated: 75\% are used for learning, while the remaining 25\% is for the validation set.   
\par For each value of SNR, three MFPLS models are computed. The two first use exclusively one dimension of the predictor: MFPLS(1) uses $X^{(1)}$ and MFPLS(2) $X^{(2)}$. The third one uses both functional components (MFPLS). The purpose is to assess the amount of performance using one-dimensional domain and multiple-dimensional domain. We also compute MFPCA-LDA for comparison purposes, the principal component analysis is performed by \cite{rMFPCA} package. \\ The number of components in the two approaches: MFPLS and MFPCA-LDA are chosen by 10-fold cross-validation using AUC. We did 200 simulations: models are assessed by AUC on the validation set.\\ 
\begin{table}[ht]
\color{black}
\centering
\begin{tabular}{rrrrr} 
\hline
\textbf{SNR} & \textbf{MFPLS} & \textbf{MFPLS(1)} & \textbf{MFPLS(2)} & \textbf{MFPCA-LDA}  \\ 
\hline
0.50& 0.93(0.03)& 0.73(0.04)& 0.80(0.04)& 0.88(0.05)                \\
0.73         & 0.95(0.02)           & 0.75(0.05)             & 0.81(0.04)               & 0.95(0.02)                \\
1.16         & 0.97(0.01)           & 0.77(0.04)              & 0.81(0.04)       & 0.97(0.02)                \\
2.10         & 0.98(0.01)           & 0.82(0.04)              & 0.80(0.04)              & 0.99(0.01)                \\
4.94         & 1.00(0.01)           & 0.85(0.04)              & 0.81(0.04)              & 0.99(0.01)                \\
\hline
\end{tabular}
\caption{\color{black} Means and the standard deviations (in parentheses) of the obtained AUC on the $200$ experiences. }
\label{IM_PLS}
\end{table}

\par 
{\color{black} Table \ref{IM_PLS} shows that MFPLS gives better results than MFPCA-LDA for the lowest value of SNR,  and the difference between the methods disappears with the increase of  SNR.} Using partially the data (models MFPLS(1) and MFPLS(2)) to predict the class variable is less efficient than using both dimensions. Namely, Table \ref{IM_PLS} clearly shows the advantage of using both of the components of the functional variables.\par This simulation demonstrates the ability of our method to classify different domain data. In addition, as it's specially designed for supervised learning, it can be more effective than principal component analysis-based techniques such as MFPCA-LDA in a noisy context.   

\section{Real data application: Multivariate time series classification}
\label{appli}
In this section, we compare the proposed methods with black box models (LSTM, Random Forest, etc...) on benchmark data (Table \ref{tab_1}, from Table 1 of \citep{multivariate}), ranging from online character recognition to activity recognition.  These data, suitable for multivariate functional time series data and binary classification, have been used by various works to assess new methodologies (see e.g. \cite{m_2}, \cite{ECG}). 

\begin{table}[ht]
    \centering
    \resizebox{\textwidth}{!}{
    \begin{tabular}{l | c | r|  c| c| l }
    \textbf{Dataset} &\textit{\textbf{d} }& \textbf{T} & \textbf{Task} & \textbf{Ratio} &  \textbf{Sources} \\ \hline
    CMUsubject16 & 62 & 534 & Action Recognition & 50-50 split & \cite{carn}    \\
    ECG & 2 & 147 & ECG Classification & 50-50 split  &  \cite{rt}\\
    EEG  & 13 & 117 & EEG Classification & 50-50 split & \cite{40}\\ 
EEG2 & 64 & 256 & EEG Classification & 20-80 split & \citep{40}\\ 
    KickvsPunch & 62 & 761 & Action Recognition & 62-38 split & \cite{carn}\\
Movement AAL  & 4 & 119 & Movement Classification & 50-50 split &\cite{40} \\
NetFlow & 4 & 994 & Action Recognition  & 60-40 split & \cite{subakan} \\ 
Occupancy  & 5 & 3758 & Occupancy Classification & 35-65 split  & \cite{40}\\
Ozone & 72 & 291 & Weather Classification & 50-50 split& \cite{40}  \\
    Wafer & 6 & 198 & Manufacturing
Classification & 25-75 split  &\cite{rt} \\ 
WalkVsRun & 62 & 1918 & Action Recognition & 64-36 split  & \cite{carn} 
    \end{tabular}
    }
    \caption{ Datasets summary.
    \textbf{T} denotes the number of sampling time points, \textbf{d}: the data dimensions, and \textbf{Ratio} of the train-test split. All datasets are available in 
    \url{https://github.com/titu1994/MLSTM-FCN/releases/tag/v1.0}}
   \label{tab_1}
\end{table}

\par The proposed models (MFPLS, TMFPLS) are compared with
 discriminant analysis (MFPCA-LDA) on scores obtained  by  Multivariate functional principal component analysis \citep{rMFPCA}, and non-functional models; the Long
Short-Term Memory Fully Convolutional Network (LSTM-FCN) and Attention LSTM-FCN (ALSTM-FCN), proposed by \cite{multivariate}. We also present the benchmark of these last models named SOTA, which gives the best performances among Dynamic time warping (DTW), Random Forest (RF), SVM with a linear kernel, SVM with a 3rd-degree polynomial kernel (SVM-Poly), and other state-of-the-art methods (see \cite{multivariate} for more details). \\
The challenge is to show that our models based on regression can be competitive. The splitting of the data into training and test samples (see Table \ref{tab_1}) is that of \cite{multivariate}. The different models mentioned above are compared by the accuracy metric, the rate of well-predicted classes obtained on the test datasets.
\\
\subsection{Choice of hyperparameters}
As for some datasets (CMUsubject, KickVsPunch, etc...), the sample size is small  (less than $50$ observations, see Table \ref{res_table}) the number of components in MFPLS and MFPCA is chosen by 20-fold cross-validation (contrary to $10$-fold in previous parts).
The maximum tree depth $m^*$ is an important hyperparameter. It may significantly affect the performance of our tree-based model, as it helps {\color{black}to prevent the overfitting of TMFPLS}.  We estimate $m^*$ by cross-validation alike procedure. More precisely, we randomly take  75\% of learning data to train an intermediate TMFPLS and 25\% for pruning. This procedure is repeated  $10$ times and let $\hat{m}^*$ be the most occurred number from these $10$. The final tree is then trained on the whole learning data, with the maximum tree depth fixed to $\hat{m}^*$. As in the previous section, group are defined as $\mathcal{G}_1=1, \ldots, \mathcal{G}_d=d, \mathcal{G}_{d+1}=\{1, ..., d \} $, to see whether FPLS gives better splitting than MFPLS. Testing several combinations of dimensions takes time, the ideal choice of groups would be guided by some prior knowledge of the data structure. \\ Two strategies are used for the number of components in the decision  tree:  \textbf{TMFPLS H-1} denotes the decision tree where only one component in MFPLS is used, and \textbf{TMFPLS H-CV} is the decision tree where the number of components is estimated by 20-fold cross-validation as in MFPLS. The first tree is faster to train than the second one, and it's less likely to overfit the data. However, the second one is expected to be a more efficient model, since it is able to estimate more complex coefficient functions  $\beta$.\par 
For all functional data methods, we use 30 B-Splines basis functions by dimension to  have a functional representation \citep{ramsay2008} of each dataset ({\color{black}see 
Figure \ref{f_18} in Appendix \ref{data} for the smoothed functions}).  This number of basis functions is chosen arbitrarily small compared to the minimum number of discrete time points ($117$) of the original raw datasets.
\subsection{Results}
Table \ref{res_table} shows that, in most cases, our models (MFPLS, TMFPLS) and MFPCA-LDA are competitive with that of \cite{multivariate} and SOTA.
In about half of the cases, TMFPLS or MFPLS reach the highest or the second-highest accuracy. TMFPLS is generally more performing than MFPLS. Note also that  MFPCA-LDA is competitive with the proposed methodologies.  The main difference between MFPCA-LDA and MFPLS is that for the first one components are searched with no regard to the response variable $Y$. \\ 
For the KickVsPunch dataset, the performance of TMFPLS H-1 is better than the one by TMFPLS H-CV. This is because TMFPLS H-CV could easily overfit when the training sample is small ($\text{N}_{\text{Train}}<20$). This is one of the well-known drawbacks of the decision tree. Tuning hyperparameters is then crucial and may have a huge impact on performances. \\
\begin{table}[ht]
\centering
\captionsetup{justification=centering}
\resizebox{\textwidth}{!}{\begin{tabular}{l | l|  l|  c|  c|  c|  c|  c|  c l}
  \hline
\textbf{Datasets} & \rotatebox[]{270}{$\textbf{N}_{\textbf{Train}}$} & \rotatebox[]{270}{$\textbf{N}_{\textbf{Test}}$}  &   \textbf{MFPLS} & \textbf{TMFPLS H-1}  &\textbf{TMFPLS H-CV} & \textbf{MFPCA-LDA} & \citeauthor{multivariate} & \textbf{SOTA} & \rotatebox[]{270}{\textbf{Methods}}\\
  \hline
 CMUsubject16 & 29 & 29&   86.21 &  \underline{89.66}  & \textbf{100} & \underline{89.66} &   \textbf{100} &   \textbf{100} & [1]  \\ 
 ECG & 100 & 100& 85  & 83  & 87 &    \underline{88} & 86 &  \textbf{93}& [2]\ \\ 
 EEG &  64 & 64&  48.44 & 54.69 & 53.12 & 46.88 & \textbf{65.63} &  \underline{62.5}    & [3]  \\ 
  EEG2 & 600 & 600 &81.83 & 68.67 &    \underline{82.67} & 72.17&  \textbf{91.33}  & 77.5 & [3] \\ 
  KickVsPunch & 16 & 10  & \underline{90} & \underline{90} & 60 & 80 &   \textbf{100} &  \textbf{100} &[2]\\ 
  MovementAAL & 157& 157& \underline{67.52} & 56.69 & 53.50 & 61.78&  \textbf{79.63} &    65.61 & [4] \\ 
  NetFlow & 803 & 534 &  84.64 & 86.52 & 85.77 & 80.90 &    \underline{95} &  \textbf{98} & [2] \\ 
  Occupancy & 41 & 76&  71.05 & 61.84 & 59.21 &  \textbf{80.26}&   \underline{76}  &  67.11& [4] \\ 
  Ozone & 173& 173& 73.99 & 73.41 & 73.41 &    \underline{79.19} & \textbf{81.5}   & 75.14 &[5]\\ 
Wafer & 298 & 896& 85.04 & 87.39 &    \underline{97.99} & 97.32  &  \textbf{99} &  \textbf{99} &  [2]\\  
WalkVsRun &28 & 16 & \textbf{100} &  \textbf{100} & \textbf{100} &  \textbf{100}&  \textbf{100 } &  \textbf{100} & [2] \\
   \hline
\end{tabular}}
\caption{ Comparison of   MFPLS, TMFPLS, and other non-FDA classification methods by their accuracies (\%) in the test set. Bold metrics
correspond to the best accuracy for each dataset and underline indicate
the second best. \\ \textbf{[1]}: \cite{cmu}, \textbf{[2]}:\cite{ECG}, \textbf{[3]}:RF, \textbf{[4]}:  SVM-Poly, \textbf{[5]}: DTW }

\label{res_table}
\end{table}
\newpage

\section{Conclusion and discussion}
\label{conc}
Statistical learning of multivariate functional data evolving in complex spaces leads to challenging questions that need the development of new methods and techniques.
In this paper, we are interested in some of these methods in the case of different functional domain settings. Namely, we propose least squares regression and classification models for multivariate functional predictors. The first classification model relies on the partial least square (PLS) regression (MFPLS) while the second one (TMFPLS) combines  PLS with a decision tree. Technical arguments on the PLS methods are given.   The finite sample performance of the regression and classification models are assessed by simulations and real data (EEG, Ozone, wafer,...) applications where we compare the proposed methods with some benchmarks, in particular
a PLS regression model of the literature (\cite{MFPLS2022}) and well known principal components regression and some machine learning models.  \\
A main specificity of our proposed models is that the multivariate functional data  considered are defined on different domains compared to the literature. This allows dealing with heterogeneous types of data (e.g. images, time series, etc.) with a potentially large number of applications as shown by the given classification case study on images and functional time series. We also give a relationship between the partial least square of multivariate functional data with its univariate counterparts. To the best of our knowledge, the proposed tree classification model is new. 

The finite sample properties show our models' competitiveness with regard to some existing methods. The multivariate time series classification case study highlights the competitive performance of MFPLS and TMFPLS with black-box models (LSTM, RF,...) on benchmark data. These performances may be improved by using prior knowledge of the benchmark data (groups of variables, suitable preprocessing, ...).\\
In this paper we focus on continuous functional predictors, a possible extension of the proposed models would be including additional type (e.g, qualitative) of covariates.\\
The EEG and ozone data considered in the finite sample study may have spatial dependence. The  classification approaches seem not affected by these data dependencies, but this deserves future investigation.\\
As in a number of functional data analysis, a tuning parameter related to the number of basis functions used to smooth the raw data or reduce the dimension of the functional space, has to be selected. In this paper, we fix or use a cross-validation approach for the choice of this parameter. Other alternatives may be based on bootstrap methods or criteria like AIC, BIC.

This work highlights the good behavior of TMFPLS and a way to deal with non-linearity in classification problems of multivariate functional data. However, with heterogeneous high-dimensional data, tree-based methods may be challenging. An alternative method could be cluster-wise regression techniques by extending the univariate case studied by \cite{PLS_c} to our context. Some other methods as lasso classification techniques can also be explored (see e.g \cite{godwin2013}). 


{\color{red}}
\backmatter
\bmhead{Supplementary information}{
\color{black}The supplementary material includes  additional figures related to the numerical experiments. }












\begin{appendices}

\section{Technical arguments}
\begin{proof}[Proof of Proposition \ref{prop_mul}] Here C-S (1) and C-S (2)  stand respectively for Cauchy-Schwartz inequality on integrals and sums. 
\begin{align*}
    \text{Cov}^2(\langle \langle X, w \rangle \rangle, Y) &=\mathbb{E}^2 \left (\langle \langle X,w \rangle \rangle Y \right)\\
&=\left[\sum_{j=1}^d\left[\int_{\mathcal{I}_j}\mathbb{E}\left( X^{(j)}(t) Y\right)w^{(j)}(t)dt \right]\right]^2    \\
     \text{ \scriptsize C-S(1)}\implies \text{Cov}^2(\langle \langle X, w \rangle \rangle, Y) & {\leq} \left[\sum_{j=1}^d\left(\int_{\mathcal{I}_j}\mathbb{E}^2(X^{(j)}(t)Y)dt\right)^{1/2}  \right. \\
     & \indent \indent 
     \left. \left(\int_{\mathcal{I}_j}[w^{(j)}(t)]^2dt\right)^{1/2}\right]^2
     \\ 
\text{ \scriptsize C-S (2) }\implies \text{Cov}^2(\langle \langle X, w \rangle \rangle, Y)& \leq \left[\sum_{j=1}^d\int_{\mathcal{I}_j} \mathbb{E}^2(X^{(j)}(t)Y)dt\right] \underbrace{\left[\sum_{j=1}^d\int_{\mathcal{I}_j}[w^{(j)}(t)]^2dt\right]}_{\vert \vert \vert w\vert \vert \vert ^2=1}\\ 
   \text{Cov}^2(\langle \langle X, w \rangle \rangle, Y) & \leq \sum_{j=1}^d\int_{\mathcal{I}_j}\mathbb{E}^2(X^{(j)}(t)Y)dt 
\end{align*}
The C-S inequalities become equalities, meaning the maximums are reached, if for $j=1, \ldots, d$ there exist non-null scalars $a$ and $a'$ such as:  
\begin{itemize}
    \item $\displaystyle w^{(j)}(t) =a \mathbb{E}(X^{(j)}(t)Y), t\in \mathcal{I}_j $
    \item  $\displaystyle \left (\int_{\mathcal{I}_j}[w^{(j)}(t)]^2dt\right)^{1/2}$ $=a'\left (\int_{\mathcal{I}_j}\mathbb{E}^2( X^{(j)}(t)Y)dt\right)^{1/2}$.
\end{itemize}
 The first condition implies the second one, indeed if $\displaystyle w^{(j)}(t)=a\mathbb{E}(X^{(j)}(t)Y)$ then  $\displaystyle \left(\int_{\mathcal{I}_j}[w^{(j)}(t)]^2dt\right)^{1/2}= \vert a\vert \left(\int_{\mathcal{I}_j}\mathbb{E}^2(X^{(j)}(t)Y)dt\right)^{1/2}$, hence $a'=\vert a\vert $. \\ 
To have $\vert \vert \vert w\vert \vert \vert =1$, we take $a=\displaystyle \left(\sum_{j=1}^p\int_{\mathcal{I}_j}\mathbb{E}^2(X^{(j)}(t)Y)dt\right)^{-1/2} $. \\  Thus, the solution of \eqref{f_1} is   
\begin{equation}
    w^{(j)}(t)=\frac{\mathbb{E}(X^{(i)}(t)Y)}{\left(\sum_{j=1}^p\int_{\mathcal{I}_j}\mathbb{E}^2(X^{(j)}(t)Y)dt\right)^{1/2}}, \; t \in \mathcal{I}_j. 
\end{equation}
\end{proof}

\begin{proof}[Proof of Proposition \ref{prop_m}]  $X$  first order residual definition is $X = \xi_1 \rho_1 + X_{1}$, where $X_{1}$ holds 
\begin{equation}
    \mathbb{E}(\xi_1X_1 ) = 0_{\mathbb{R}^d} \iff  \mathbb{E}(\xi_1X_1^{(j)} (t) )= 0 \; t\in \mathcal{I}_j, \;  1 \leq j \leq d. 
    \label{orth} 
\end{equation}
Analogously higher-order residuals also verify
\begin{equation}
    \mathbb{E}(\xi_h X_h)= 0_{\mathbb{R}^{d} } 
    \label{h_orthog} \text{  }\forall h \in \mathbb{N}^*. \\
\end{equation}
To show that $\{\xi_k\}_{k=1}^h$ forms an orthogonal system, we use a proof by induction, similarly to  \cite{PLS1995}.\\ \\
The base case verifies. Indeed, \eqref{orth} implies that 
$$
\mathbb{E}(\xi_1 \xi_2)= \sum_{j=1}^{d} \int_{\mathcal{I}_j } \mathbb{E}\left(\xi_1 X_{1}^{(j)}(t)\right) w_2^{(j)}(t)dt= 0.
$$
Assume the induction hypothesis $\mathcal{H}_0$,  $\mathcal{H}_0$: $\{ \xi_k\}_{k=1}^{h} $ \textit{forms an orthogonal system } $h\geq 1$  
\begin{align*}
    \mathbb{E}(\xi_{h}
    \xi_{h+1}) &=   \sum_{j=1}^{d} \int_{\mathcal{I}_j} \mathbb{E}\left(\xi_h X_{h}^{(j)}(t)\right) w_{h+1}^{(j)}(t)dt \\ 
    (\ref{h_orthog}) \implies \mathbb{E}(\xi_{h}\xi_{h+1}) &= 0 \\
    \mathbb{E}(\xi_{h-1}\xi_{h+1})& = \sum_{j=1}^{d}  \int_{\mathcal{I}_j} \mathbb{E}\left(\xi_{h-1} X_{h}^{(j)}(t)\right) w_{h+1}^{(j)}(t)dt \\ 
 \text{Since } X_{h-1}  =  \rho_{h}\xi_h+X_h \\  \implies \mathbb{E}(\xi_{h-1}\xi_{h+1}) &=
 \sum_{j=1}^{d}  \int_{\mathcal{I}_j } \underbrace{\mathbb{E}\left(\xi_{h-1}X_{h-1}^{(j)}(t)\right)}_{= 0 \text{ by } (\ref{h_orthog})}dt\\ &\; \; -\rho_{h}^{(j)}(t) \int_{\mathcal{I}_j } \underbrace{\mathbb{E}\left( \xi_{h-1}\xi_h \right)}_{= 0 \text{ by } \mathcal{H}_0}\sum_{j=1}^{d}w_{h+1}^{(j)}(t) dt,
 \\
\text{then } \mathbb{E}(\xi_{h-1}\xi_{h+1})  &= 0
\end{align*}
The same procedure can be used to show that  $\mathbb{E}(\xi_j \xi_{h+1} )$ = 0 $\forall j \leq h-2$. Hence, $\{\xi_k\}_{k=1}^h $ forms an orthogonal system $\forall h\geq 1$.\\ 
The expansion formulas are implications of this point.
\end{proof}
\begin{proof}[Proof of Lemma \ref{lemm_multi} ]
For $h =1$, we have  $v_1= w_1$, as $\xi_1=\langle \langle X, w_1\rangle \rangle $,  the base case verifies.\\
Assume that $\langle \langle X, v_{ j} \rangle\rangle = \xi_j$ is true up to order $h$ ($\forall j \leq h$). \\
Recall that, 
\begin{equation}
    \xi_{h+1} = \langle \langle X_{h}, w_{h+1}\rangle \rangle .
\end{equation}
The second equation of  Proposition  \ref{prop_m}, gives that $$X_{h}=  X - \sum_{i=1}^{h} \rho_i \langle \langle v_i, X \rangle \rangle$$.\\ Then   
\begin{equation*}
\begin{split}
  \xi_{h+1}  = \langle \langle X, w_{h+1} \rangle \rangle   - \sum_{i=1}^{h} \langle \langle v_i,  X \rangle \rangle  \langle \langle \rho_i, w_{h+1} \rangle \rangle  = \langle \langle X,  \underbrace{w_{h+1} - \sum_{i=1}^h\langle \langle  \rho_i,  w_{h+1}\rangle  \rangle  v_i }_{v_{h+1} }\rangle \rangle 
\end{split}
\end{equation*}
This concludes the proof. 
\\ 
\end{proof}

\section{Additional figures}
\label{data}
\begin{figure}[ht]
    \centering
     \resizebox{\textwidth}{!}{
\begin{tabular}{l |c |c}
\centering
\textbf{Dataset} & \textbf{Raw} & \textbf{Smoothed}\\ \hline \\ 
ECG&  \includegraphics[align=c, scale=.13]{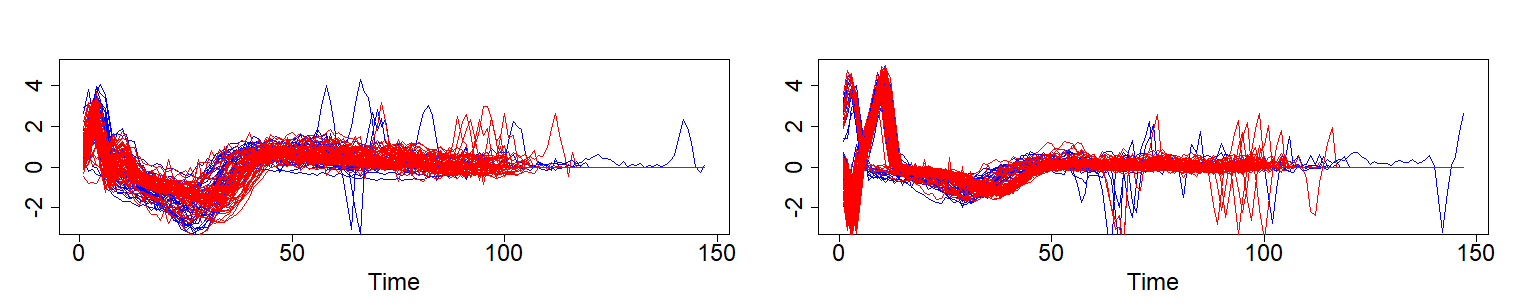}& \includegraphics[align=c, scale=.13]{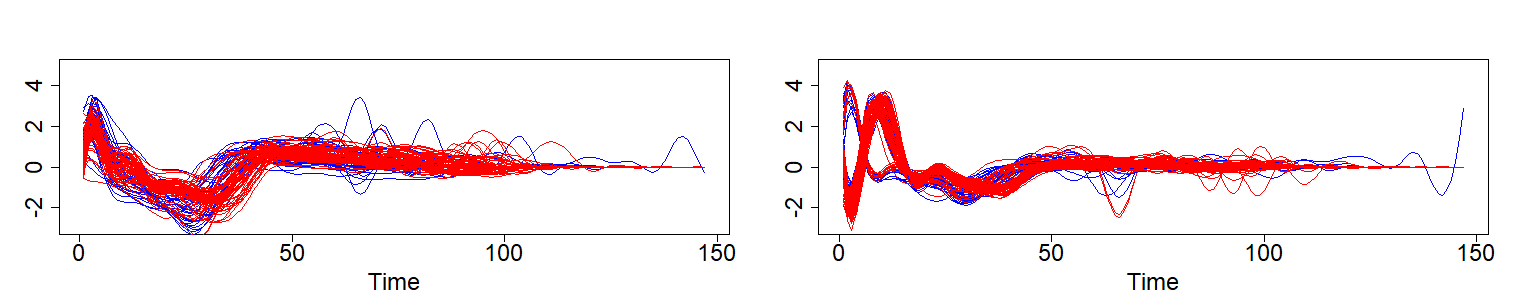} \\
    Netflow &\includegraphics[align=c,scale=.13]{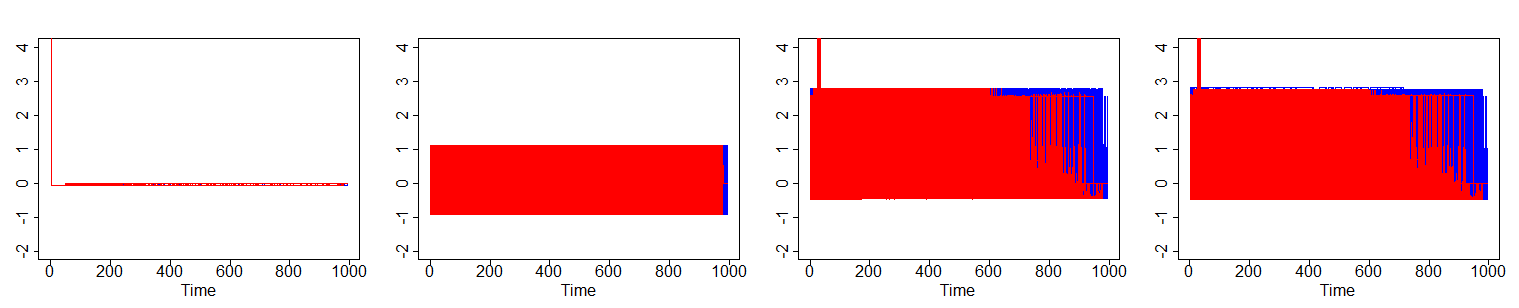}&  \includegraphics[align=c,scale=.13]{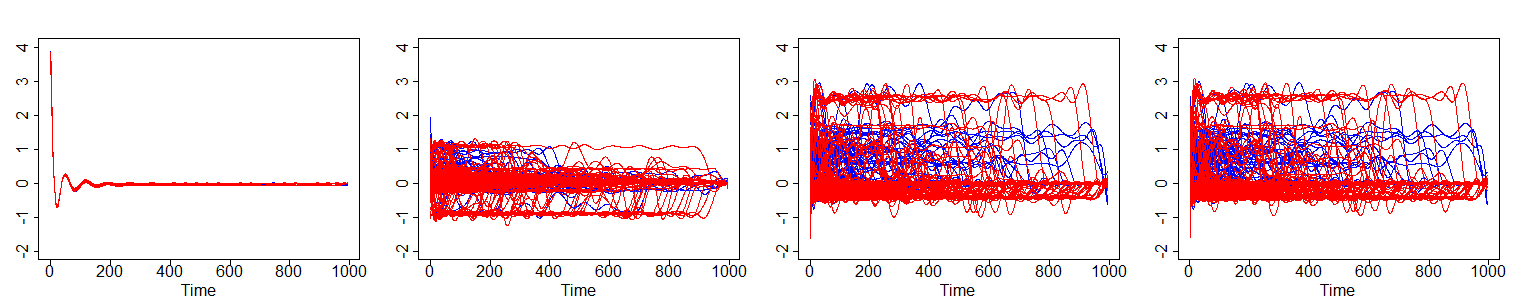} \\ 
    Occupancy &\includegraphics[align=c, scale=.13]{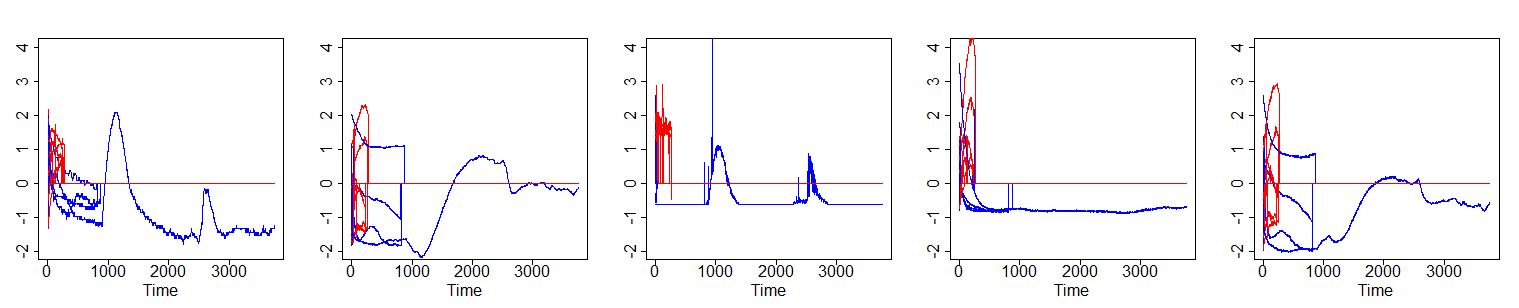} & \includegraphics[align=c,scale=.13]{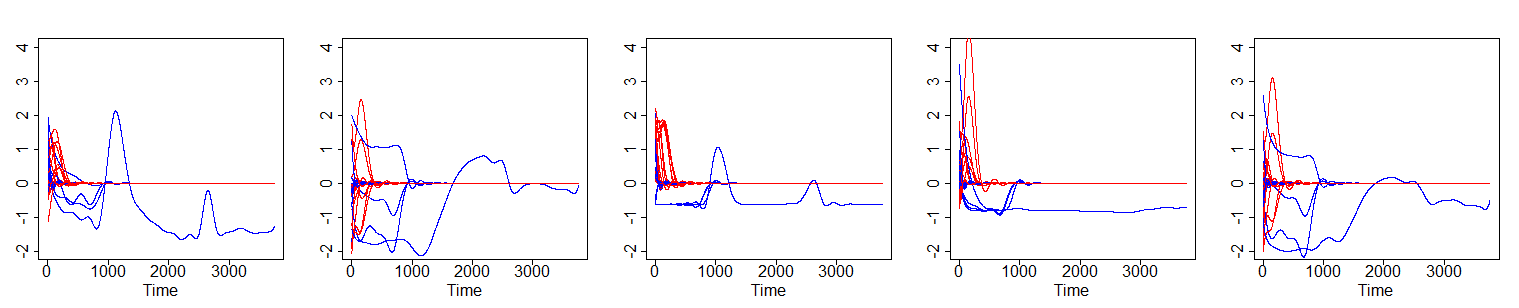}\\     
    Movement All &\includegraphics[align=c, scale=.13]{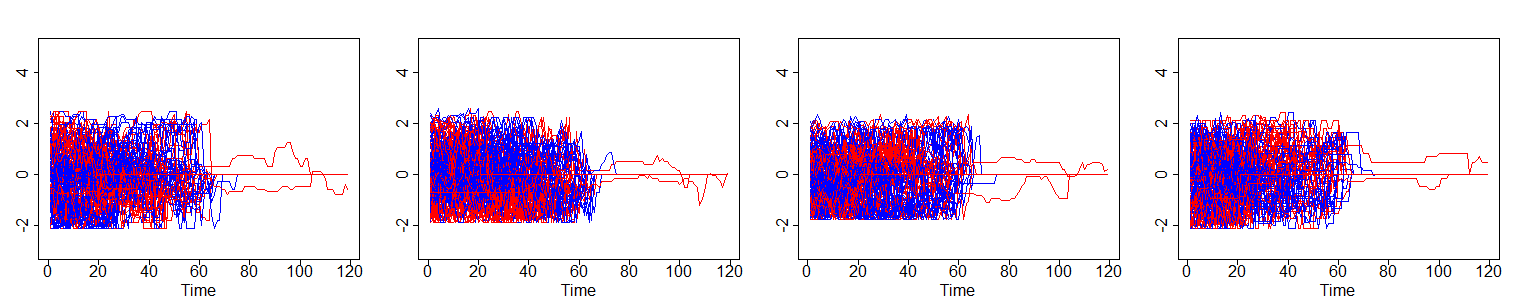} & \includegraphics[align=c,scale=.13]{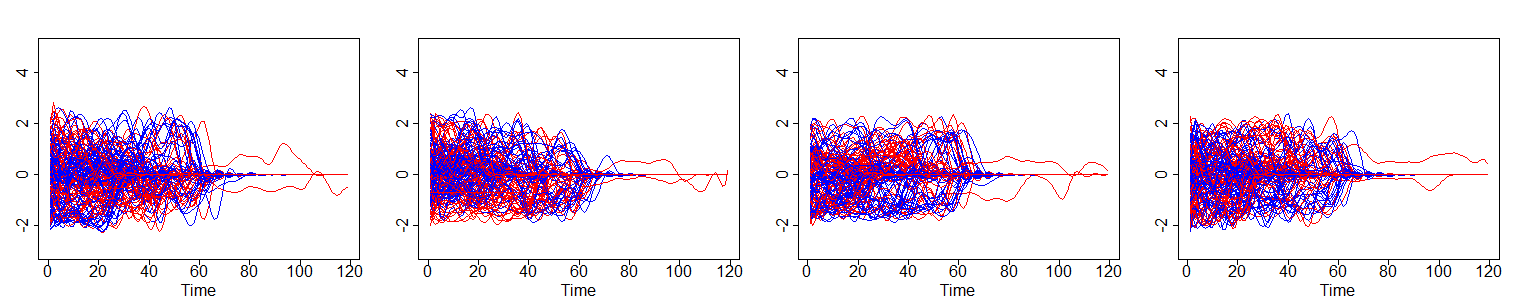}\\ 
     Wafer & \includegraphics[align=c, scale=.2]{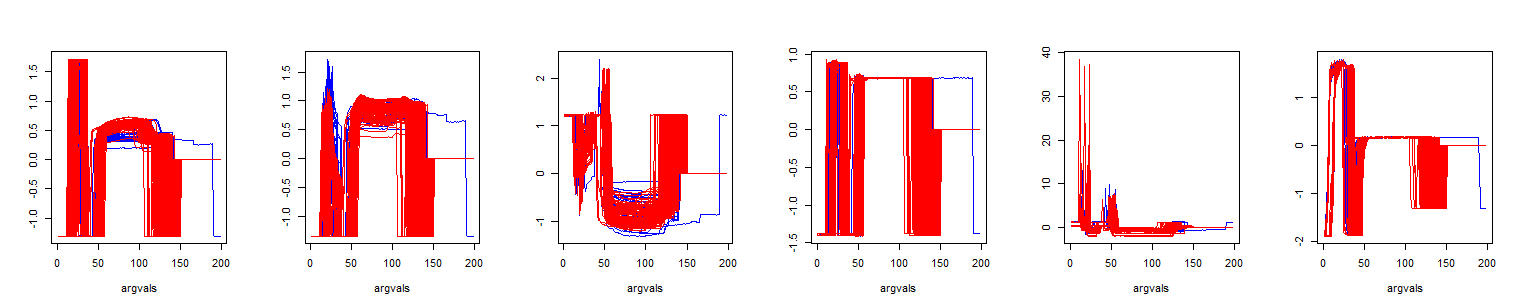} &  \includegraphics[align=c, scale=.2]{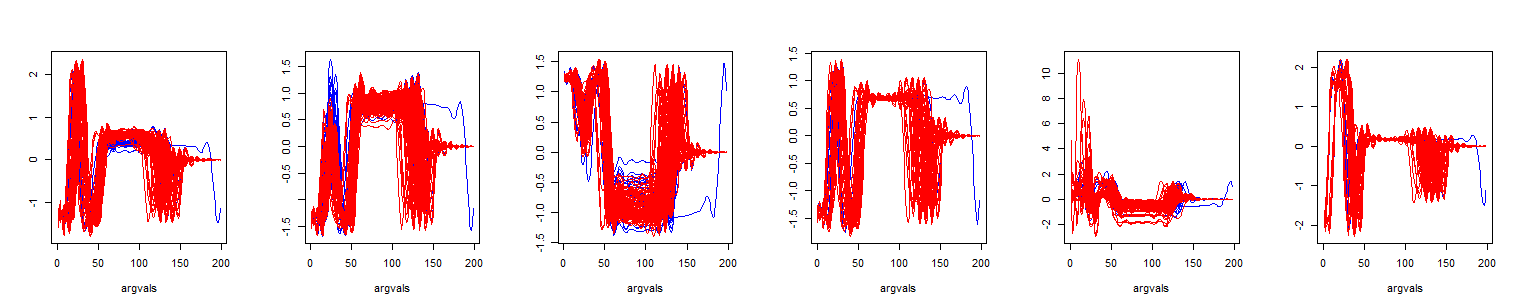}
\end{tabular}
}
\caption{ \color{black} Presentation of some of the benchmark datasets (those with less than 10 dimensions). The smoothing is done using a basis of $30$ splines functions for each dimension.}
\label{f_18}
\end{figure}

\end{appendices}


\bibliography{ref.bib}



\end{document}